\documentclass[twocolumn]{svjour3}          
\smartqed  
\usepackage{graphicx}
\usepackage{enumerate}
\usepackage{latexsym}
\usepackage{amsfonts}
\usepackage{amsmath}
\usepackage{amssymb}
\usepackage[bigsqcap]{stmaryrd}
\usepackage{color}
\usepackage{colortbl}
\usepackage{epsfig}
\usepackage{xspace}
\usepackage{graphicx}
\usepackage{esvect} 
\usepackage{subfigure}
\usepackage{balance}
\usepackage{cite}
\usepackage[english]{babel}

\usepackage{multirow}
\usepackage{url}

\newcommand{\at}[1]{\protect\ensuremath{\mathsf{#1}}\xspace}
\newcommand{\myhrule}{\rule[.5pt]{\hsize}{.5pt}}

\newcommand{\eat}[1]{}

\newcommand{\sstab}{\vspace{0.5ex}\noindent}

\newcommand{\eg}{\emph{e.g.,}\xspace}
\newcommand{\ie}{\emph{i.e.,}\xspace}

\newcommand{\wrt}{\emph{w.r.t.}\xspace}

\newcommand{\SET}{\mbox{{\bf set}}\ }

\newcommand{\If}{\mbox{\bf if}\ }

\newcommand{\Then}{\mbox{\bf then}\ }

\newcommand{\Else}{\mbox{\bf else}\ }
\newcommand{\ElseIf}{\mbox{\bf elseif}\ }
\newcommand{\While}{\mbox{\bf while}\ }

\newcommand{\Do}{\mbox{\bf do}\ }

\newcommand{\Repeat}{\mbox{\bf repeat}\ }
\newcommand{\Until}{\mbox{\bf until}\ }
\newcommand{\For}{\mbox{\bf for}\ }

\newcommand{\Handle}{\mbox{\bf handle}\ }

\newcommand{\Or}{\mbox{\bf or}\ }
\renewcommand{\And}{\mbox{\bf and}\ }

\newcommand{\Return}{\mbox{\bf return}\ }

\newcommand{\kw}[1]{{\ensuremath {\mathsf{#1}}}\xspace}

\newcounter{ccc}
\newcommand{\bcc}{\setcounter{ccc}{1}\theccc.}
\newcommand{\icc}{\addtocounter{ccc}{1}\theccc.}

\newcommand{\eop}{\hspace*{\fill}\mbox{$\Box$}}     

\renewcommand{\ni}{\noindent}

\newcommand{\nthesection}{\arabic{section}}

\newcounter{prop}
\renewcommand{\theprop}{\arabic{theorem}}

\newcounter{cor}
\renewcommand{\thecor}{\arabic{theorem}}
\newenvironment{prop}{\begin{em}
        \refstepcounter{theorem}
        {\vspace{1.5ex}\noindent \bf Proposition \theprop:}}{
        \end{em}\eop\vspace{1.5ex}}
\newcounter{alg}[section]
\renewcommand{\thealg}{\nthesection.\arabic{alg}}

\newcounter{arule}
\renewcommand{\thearule}{\arabic{arule}}

\newcommand{\ei}{\end{itemize}}
\newcommand{\ee}{\end{enumerate}}
\newcommand{\mat}[2]{{\begin{tabbing}\hspace{#1}\=\+\kill #2\end{tabbing}}}

\newcommand{\beqn}{\vspace{-1ex}\begin{eqnarray*}}
\newcommand{\eeqn}{\vspace{-1ex}\end{eqnarray*}}

\newcommand{\stitle}[1]{\vspace{0.5ex} \noindent{\bf #1}}

\renewcommand{\ni}{\noindent}

\renewcommand{\texttt}[1]{{\small\textsf{#1}}}

\newcommand{\kwlog}{\emph{w.l.o.g.}\xspace}
\newcommand{\lsa}{\kw{LS}}
\newcommand{\dpa}{\kw{DP}}
\newcommand{\dps}{\kw{DPSED}}

\newcommand{\bqsa}{\kw{BQS}}
\newcommand{\fbqsa}{\kw{FBQS}}
\newcommand{\operb}{\kw{OPERB}}

\newcommand{\squishe}{\kw{SQUISH}-\kw{E}}
\newcommand{\sleeve}{\kw{Sleeve}}
\newcommand{\pavlidis}{\kw{Theo~Pavlidis'}}

\newcommand{\cia}{\kw{SI}}
\newcommand{\cpia}{\kw{CPolyInter}} 
\newcommand{\rpia}{\kw{FastRPolyInter}} 

\newcommand{\cist}{\kw{CISED}-\kw{S}}
\newcommand{\cista}{\kw{CISED}-\kw{W}}

\newcommand{\dpsed}{\kw{DPSED}}

\newcommand{\truck}{\kw{Truck}}
\newcommand{\mopsi}{\kw{Mopsi}}
\newcommand{\sercar}{\kw{ServiceCar}}
\newcommand{\pricar}{\kw{PrivateCar}}
\newcommand{\geolife}{\kw{GeoLife}}

\newcommand{\trajec}[1]{$\dddot{\mathcal{#1}}$}

\newcommand{\ped}{\kw{PED}} 
\newcommand{\sed}{\kw{SED}} 
\newcommand{\ded}{\kw{DED}} 
\newcommand{\sector}[1]{{$\mathcal{S}{#1}$}}
\newcommand{\cone}[1]{{$\mathcal{C}{#1}$}}
\renewcommand{\circle}[1]{{$\mathcal{O}{#1}$}}
\newcommand{\pcircle}[1]{{$\mathcal{O}^c{#1}$}}

\begin{document}

\title{One-Pass Trajectory Simplification Using the Synchronous Euclidean Distance}

\author{Xuelian~Lin \and
	Jiahao~Jiang \and
	Shuai~Ma \and
	Yimeng~Zuo \and
	Chunming~Hu
}

\institute{X. Lin, J. Jiang, S. Ma (correspondence), Y. Zuo and C. Hu \at
	SKLSDE lab, School of Computer Science and Engineering, Beihang University, China. \\
	\email{\{linxl, jiangjh, mashuai, zuoym, hucm\}@buaa.edu.cn}           
}

\date{Received: xxx, 2017 / Accepted: xxx, 2018}

\maketitle

\begin{abstract}
Various mobile devices have been used to collect, store and transmit tremendous trajectory data, and it is known that raw trajectory data seriously wastes the storage, network band and computing resource.  To attack this issue, one-pass  line simplification (\lsa) algorithms have are been developed, by compressing data points in a trajectory to a set of continuous line segments. However, these algorithms adopt the {\em perpendicular Euclidean distance}, and none of them uses the {\em synchronous Euclidean distance} (\sed), and cannot support spatio-temporal queries.
%
%
To do this, we develop two one-pass error bounded trajectory simplification algorithms (\cist and \cista) using \sed,
based on a novel spatio-temporal cone intersection technique.
Using four real-life trajectory datasets, we experimentally show that our approaches are both efficient and effective.
In terms of running time, algorithms \cist and \cista are on average {$3$ times faster} than \squishe (the most efficient existing \lsa algorithm using \sed). In terms of compression ratios, algorithms \cist and \cista are comparable with and {$19.6\%$} better than \dps (the most effective existing \lsa algorithm using \sed) on average, respectively, and are {$21.1\%$} and {$42.4\%$} better than \squishe on average, respectively.
\end{abstract}

\section{Introduction}
\label{sec-intro}

Various mobile devices, such as smart-phones, on-board diagnostics, personal navigation devices, and wearable smart devices, have been using their sensors to collect massive trajectory data of moving objects at a certain sampling rate (e.g., a data point every $5$ seconds), which is transmitted to cloud servers for various applications such as location based services and trajectory mining.
Transmitting and storing raw trajectory data consumes too much network bandwidth and storage capacity \cite{Chen:Trajectory, Meratnia:Spatiotemporal,Shi:Survey, Lin:Operb, Liu:BQS, Liu:Amnesic, Muckell:survey, Muckell:Compression,Cao:Spatio, Popa:Spatio,Nibali:Trajic}. 
%
%
It is known that these issues can be resolved or greatly alleviated by trajectory compression techniques via removing redundant data points of trajectories \cite{Douglas:Peucker, Hershberger:Speeding, Meratnia:Spatiotemporal,Lin:Operb, Liu:BQS, Liu:Amnesic,  Muckell:Compression, Chen:Trajectory, Cao:Spatio,  Nibali:Trajic, Long:Direction, Popa:Spatio, Han:Compress, Chen:Fast}, among which the piece-wise line simplification technique is widely used \cite{Douglas:Peucker, Meratnia:Spatiotemporal,  Muckell:Compression, Chen:Trajectory, Cao:Spatio, Liu:BQS, Liu:Amnesic, Lin:Operb, Chen:Fast}, due to its distinct advantages: (a) simple and easy to implement, (b) no need of extra knowledge and suitable for freely  moving  objects, and (c) bounded errors with good compression ratios \cite{Popa:Spatio,Lin:Operb}.

\begin{figure*}[tb!]
\centering
\includegraphics[scale=0.78]{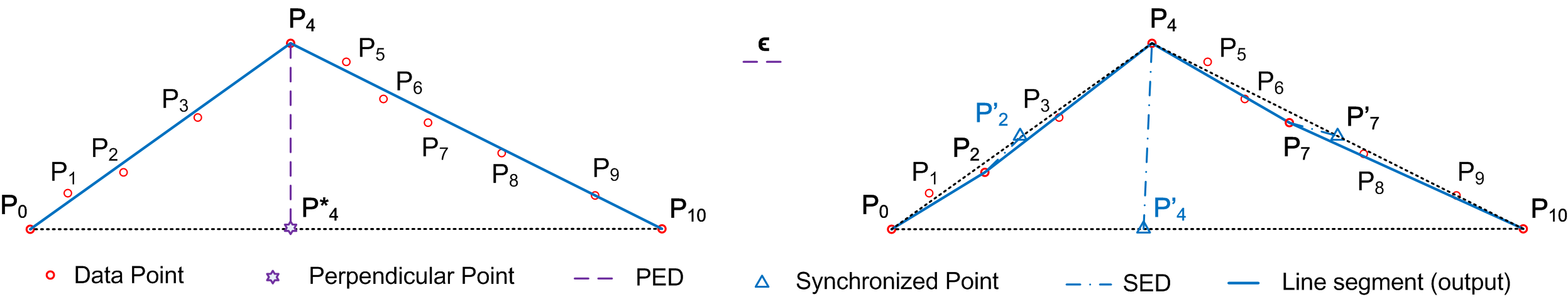}
\caption{\small A trajectory $\dddot{\mathcal{T}}[P_0, \ldots, P_{10}]$  with eleven points is represented by two (left) and four (right) continuous line segments (solid blue), compressed by the Douglas--Peucker algorithm \cite{Douglas:Peucker} using \ped and \sed, respectively. The Douglas--Peucker algorithm firstly creates line segment $\protect\vv{P_0P_{10}}$, then it calculates the distance of each point in the trajectory to $\protect\vv{P_{0}P_{10}}$. It finds that point $P_{4}$ has the maximum distance to $\protect\vv{P_{0}P_{10}}$, and is greater than the user defined threshold $\epsilon$. Then it goes to compress sub-trajectories $[P_0, \ldots, P_{4}]$ and $[P_{4}, \ldots, P_{10}]$, separately.
}
\vspace{-1ex}
\label{fig:notations}
\end{figure*}

Originally, line simplification (\lsa) algorithms adopt the \emph{perpendicular Euclidean distance} (\ped) as a metric to compute the errors,
\eg $|\vv{P_4P^*_4}|$ is the \ped of data point $P_4$ to line segment $\vv{P_0P_{10}}$ in Figure~\ref{fig:notations} (left).
Line simplification algorithms using \ped have good compression ratios~ \cite{Douglas:Peucker, Hershberger:Speeding,Lin:Operb, Liu:BQS, Muckell:Compression, Chen:Trajectory, Cao:Spatio, Shi:Survey}.  However, when using \ped, the temporal information is lost. Thus, a spatio-temporal query, \eg ``\emph{the position of a moving object at time $t$}", on the compressed trajectories by \lsa algorithms using \ped may return an approximate point $P'$ whose distance to the actual position $P$ of the moving object at time $t$ is unbounded. 

\eat{, of a data point to a proposed generalized line (\eg in Figure~\ref{fig:notations} (left), $|P_4P^*_4|$ is the \ped of $P_4$ to the line $\overline{P_0P_{10}}$) as the condition to discard or retain that data point.
\lsa algorithms using \ped have good compression ratios and are error bounded on \ped, hence they are widely used in scenarios that compression ratio is the most concerned factor. However, when using \ped, the temporal information of trajectory points is lost. Thus, a temporal-spatio query, \eg ``\emph{the position $P$ of a moving object at time $t$}", on trajectories compressed by \lsa algorithms using \ped returns an approximated point $P'$ whose distance to the actual position $P$ at time $t$ is unbounded. For example, a query at time $t_7$ returns an approximated point $P'_7$ whose distance to the point $P_7$ is great than the threshold $\epsilon$.

,  and implemented it in Douglas-Peucker (\dpa)~\cite{Douglas:Peucker} algorithm
}

The \emph{synchronous Euclidean distance} (\sed) was then introduced for trajectory compression to support the above spatio-temporal queries \cite{Meratnia:Spatiotemporal}. \sed is the Euclidean distance of a data point to its \emph{approximate temporally synchronized data point} \cite{Meratnia:Spatiotemporal} on the corresponding line segment. For instance, $P'_4$ and $P'_7$ are the \emph{synchronized data points} of points $P_4$ and $P_7$ \wrt line segments $\vv{P_0P_{10}}$ and $\vv{P_4P_{10}}$, respectively, in Figure~\ref{fig:notations} (right).
\lsa algorithms using \sed may produce more line segments. However, the use of \sed ensures that the Euclidean distance between a data point and its  synchronized point \wrt the corresponding line segment is limited within a distance bound $\epsilon$. Hence, the above spatio-temporal query over the trajectories compressed by \sed enabled approaches returns the synchronized point $P'$ of a data point $P$ within the bound $\epsilon$.


The problem of finding the minimal number of line segments to represent the original polygonal lines \wrt an error bound $\epsilon$ is known as the ``min-\#" problem\cite{Imai:Optimal,Chan:Optimal}.
An optimal $O(n^3)$  \lsa algorithm using \ped was firstly developed in \cite{Imai:Optimal},  where $n$ is the number of the original points.
	Later, an improved optimal  $O(n^2)$  algorithm using \ped was designed in \cite{Chan:Optimal}, with the help of \textit{sector intersection} mechanism.
	However, the time complexity of the optimal \lsa algorithm using \sed remains in $O(n^3)$, as the optimization mechanisms are \ped specific, and cannot work with \sed.

Due to the high time complexities of optimal \lsa algorithms using \sed, sub-optimal \lsa algorithms using \sed have been developed for trajectory compression, including batch algorithms (\eg Douglas-Peucker based algorithm \dpsed \cite{Meratnia:Spatiotemporal}) and online algorithms (\eg\ \squishe \cite{Muckell:Compression}).
However, these methods still have high time and/or space complexities, which hinder their utilities in resource-constrained devices. 

Observe that linear time \lsa algorithms using \ped \cite{Williams:Longest, Sklansky:Cone, Dunham:Cone, Zhao:Sleeve, Lin:Operb} have been develped, and they are more efficient for resource-constrained devices.
The key idea to achieve a linear time complexity is by local distance checking in constant time, \eg the \textit{sector intersection} mechanism used in \cite{Williams:Longest, Sklansky:Cone, Dunham:Cone, Zhao:Sleeve} and the \textit{fitting function} approach used in our preview work \cite {Lin:Operb}.
Unfortunately, these techinques are designed specifically for \ped, and cannot be applied for \sed.

Indeed, it is even more challenging to design an one-pass \lsa algorithm using \sed than using \ped.
To our knowledge,  no one-pass \lsa algorithms using \sed have been developed in the community.

\stitle{{Contributions}}.
To this end, we propose two one-pass error bounded \lsa algorithms using \sed for compressing trajectories in an efficient and effective way. 

\sstab {(1)} We first develop a novel local synchronous distance checking approach, \ie spatio-temporal \underline{C}one \underline{I}ntersection using the \underline{S}ynchronous \underline{E}uclidean \underline{D}istance (CISED).
We further approximate the intersection of spatio-temporal cones with the intersection of regular polygons, and develop a fast regular polygon intersection algorithm, such that each data point in a trajectory is checked in $O(1)$ time during the entire process of trajectory simplification.


\sstab {(2)} We next develop two one-pass trajectory simplification algorithms \cist and \cista, achieving $O(n)$ time complexity and $O(1)$ space complexity, based on our local synchronous distance checking technique.
Algorithm \cist belongs to strong simplification that only has original points in its outputs, while algorithm \cista belongs to weak simplification that allows interpolated data points in its output.

\sstab (3) Using four real-life trajectory datasets (\truck, \sercar, \geolife, \pricar),
we finally conduct an extensive experimental study, by comparing our methods \cist and \cista  with the optimal \lsa algorithm using \sed, \dps~\cite{Meratnia:Spatiotemporal} (the most effective existing \lsa algorithm using \sed) and \squishe \cite{Muckell:Compression} (the most efficient existing \lsa algorithm using \sed).

For running time,
algorithms \cist and \cista are on average ($14.21$, $18.19$, $17.06$, $9.98$), ($2.84$, $3.45$, $3.69$, $2.86$) and ($925.25$, $7888.26$, $40041.59$, $8528.76$) times faster than \dps, \squishe and the optimal \lsa algorithm on datasets (\sercar, \geolife, \mopsi, \pricar), respectively.

For compression ratios, algorithm \cist is better than \squishe and comparable with \dps. The output sizes of \cist are on average ($79.5\%$, $79.5\%$, $66.0\%$, $72.7\%$),
($109.1\%$, $109.7\%$, $113.5\%$, $109.2\%$) and ($134.3\%$, $133.3\%$,
  $148.4\%$, $135.4\%$) of \squishe, \dps and the optimal \lsa algorithm on datasets (\sercar, \geolife, \mopsi, \pricar), respectively.
Moreover, algorithm \cista is comparable with the optimal \lsa algorithm and better than \squishe and \dps that are on average ($57.9\%$, $58.8\%$, $48.4\%$, $54.6\%$),
($79.6\%$, $81.2\%$, $83.3\%$, $82.1\%$) and ($98.1\%$, $98.7\%$, $108.9\%$,
  $101.7\%$) of \squishe, \dps and the optimal \lsa algorithm on datasets (\sercar, \geolife, \mopsi, \pricar), respectively.

\stitle{{Organization}}.
The remainder of the article is organized as follows.
Section \ref{sec-preliminary} introduces the basic concepts and techniques.
Section \ref{sec-localcheck} presents our local synchronous distance checking method.
Section \ref{sec-alg} presents our one-pass trajectory simplification algorithms.
Section \ref{sec-exp} reports the experimental results, followed by related work in
Section \ref{sec-related} and conclusion in Section \ref{sec-conclusion}.

\section{Preliminaries}
\label{sec-preliminary}

In this section, we first introduce basic concepts for piece-wise line based trajectory compression.
We then describe the optimal \lsa algorithm and the \textit{sector intersection} mechanism, and show how this mechanism can be used to fast the \lsa algorithms using \ped and why it cannot work with \sed.
Finally, we illustrate a convex polygon intersection algorithm, which serves as one of the fundamental components of our local synchronous distance checking method.

Notations used are summarized in Table~\ref{tab:notations}.


\subsection{Basic Notations}
\label{subsec-notation}

We first introduce basic notations.

\begin{table}
	\renewcommand{\arraystretch}{1.35}
	\caption{\small Summary of notations}
	\centering
	\footnotesize
	\begin{tabular}{|c|l|}
		\hline
		{\bf Notations}& {\bf Semantics}   \\
		\hline 
		$P$ & a data point \\
		\hline
		$\dddot{\mathcal{T}}$ & a trajectory $\dddot{\mathcal{T}}$ is a sequence of data points\\
		\hline
		$\overline{\mathcal{T}}$&  {a piece-wise line representation of a }	\\
								& trajectory $\dddot{\mathcal{T}}$		\\
		\hline
		$\mathcal{L}$ & a directed line segment  \\
		\hline
		$ped(P, \mathcal{L})$ &  {the perpendicular Euclidean distance of }	\\
								& point $P$ to line segment $\mathcal{L}$	\\
		\hline
		$sed(P, \mathcal{L})$ & {the synchronous Euclidean distance of }	\\
								& point $P$ to line segment $\mathcal{L}$	\\
		\hline
		$\epsilon$ & the error bound \\
		\hline
		\sector{} & a sector \\
		\hline
		$\vv{A} \times \vv{B}$ & the cross product of (vectors) $\vv{A}$ and $\vv{B}$\\
		\hline
		$\mathcal{H}(\mathcal{L})$ & The open half-plane to the left of $\mathcal{L}$ \\
		\hline
		$\mathcal{R}$& a convex polygon \\
		\hline
		$\mathcal{R}^*$ & the intersection of convex polygons \\
		\hline
		$m$ & the maximum number of edges of a polygon\\
		\hline
		$E^j$ & a group of edges labeled with $j$\\
		\hline
		$g(e)$ & the label of an edge $e$ of polygons \\
		\hline
		\circle{} & a synchronous circle\\
		\hline
		\cone{} & a spatio-temporal cone \\
		\hline
		\pcircle{} & a cone projection circle \\
		\hline
		$\bigsqcap$ & intersection of geometries\\
		\hline
		$G$ &	the reachability graph of a trajectory\\
		\hline
	\end{tabular}
	\label{tab:notations}
	\vspace{-1ex}
\end{table}

\stitle{Points ($P$)}. A data point is defined as a triple $P(x, y, t)$, which represents that a moving object is located at {\em longitude} $x$ and {\em latitude} $y$ at {\em time} $t$. Note that data points can be viewed as points in a three-dimension Euclidean space.

\stitle{Trajectories ($\dddot{\mathcal{T}}$)}. A trajectory $\dddot{\mathcal{T}}[P_0, \ldots, P_n]$ is a sequence of data points in a monotonically increasing order of their associated time values (\ie $P_i.t < P_j.t$ for any $0\le i<j\le n$). Intuitively, a trajectory is the path (or track) that a moving object follows through space as a function of time~\cite{physics-trajectory}.

\stitle{Directed line segments ($\mathcal{L}$)}. A directed line segment (or line segment for simplicity) $\mathcal{L}$ is defined as $\vv{P_{s}P_{e}}$, which represents the closed line segment that connects the start point $P_s$ and the end point $P_e$.
Note that here $P_s$ or $P_e$ may not be a point in a trajectory $\dddot{\mathcal{T}}$.


We also use $|\mathcal{L}|$ and $\mathcal{L}.\theta\in [0, 2\pi)$ to denote the length of a directed line segment $\mathcal{L}$, and its angle with the $x$-axis of the coordinate system $(x, y)$, where $x$ and $y$ are the longitude and latitude, respectively.
That is, a directed line segment $\mathcal{L}$ = $\vv{P_{s}P_{e}}$ can be treated as a triple $(P_s, |\mathcal{L}|, \mathcal{L}.\theta)$.

\stitle{Piecewise line representation ($\overline{\mathcal{T}}$)}. A piece-wise line representation $\overline{\mathcal{T}}[\mathcal{L}_0, \ldots , \mathcal{L}_m]$ ($0< m \le n$) of a trajectory $\dddot{\mathcal{T}}[P_0, \ldots, P_n]$ is a sequence of continuous directed line segments $\mathcal{L}_{i}$ = $\vv{P_{s_i}P_{e_i}}$ ($i\in[0,m]$) of $\dddot{\mathcal{T}}$  such that $\mathcal{L}_{0}.P_{s_0} = P_0$, $\mathcal{L}_{m}.P_{e_m} = P_n$ and  $\mathcal{L}_{i}.P_{e_i}$ = $\mathcal{L}_{i+1}.P_{s_{i+1}}$ for all $i\in[0, m-1]$. Note that each directed line segment in $\overline{\mathcal{T}}$ essentially represents a continuous sequence of data points in $\dddot{\mathcal{T}}$.

\eat{
\subsubsection{Notations of error metrics}

{For line simplification, there are distance based and shape based error metrics\cite{Shi:Survey} that measure the errors between the original trajectory and the simplified trajectory.
However, for trajectory simplification, the distance based metrics are definitely the distinct metrics.}

\stitle{Included angles ($\angle$)}. Given two directed line segments $\mathcal{L}_1$ = $\vv{P_{s}P_{e_1}}$ and $\mathcal{L}_2$ = $\vv{P_{s}P_{e_2}}$ with the same start point $P_s$, the included angle from $\mathcal{L}_1$ to $\mathcal{L}_2$, denoted as $\angle(\mathcal{L}_1, \mathcal{L}_2)$,  is $\mathcal{L}_2.\theta - \mathcal{L}_1.\theta$. For convenience, we also represent the included angle  $\angle(\mathcal{L}_1, \mathcal{L}_2)$ as $\angle{P_{e_1}P_sP_{e_2}}$.
}


\stitle{Perpendicular Euclidean Distance (\ped)}. Given a data point $P$ and a directed line segment $\mathcal{L}$ = $\vv{P_{s}P_{e}}$, the perpendicular Euclidean distance (or simply perpendicular distance) $ped(P, \mathcal{L})$ of point $P$ to line segment $\mathcal{L}$ is $\min\{|PQ|\}$ for any point $Q$ on $\vv{P_{s}P_{e}}$.

\stitle{Synchronized points \cite{Meratnia:Spatiotemporal}}. Given a sub-trajectory $\dddot{\mathcal{T}}_s[P_s$, $\ldots, P_e]$, the synchronized point $P'$ of a data point  $P \in \dddot{\mathcal{T}}_s$,~\wrt line segment $\vv{P_sP_e}$ is defined as follows:
(1) $P'.x$ = $P_s.x +  c\cdot(P_e.x - P_s.x)$,
(2) $P'.y$ = $P_s.y +  c\cdot(P_e.y - P_s.y)$ and
(3) $P'.t$ = $P.t$, where $c= \frac{P.t-P_s.t}{P_e.t-P_s.t}$.

\stitle{Synchronous Euclidean Distance (\sed) \cite{Meratnia:Spatiotemporal}}. Given a data point $P$ and a directed line segment $\mathcal{L}$ = $\vv{P_{s}P_{e}}$, the synchronous Euclidean distance (or simply synchronous distance) $sed(P, \mathcal{L})$ of $P$ to $\mathcal{L}$ is $|\vv{PP'}|$ that is the Euclidean distance from $P$ to its synchronized data point $P'$ \wrt $\mathcal{L}$. 

We illustrate these notations with examples.

\begin{example}
\label{exm-notations}
Consider Figure~\ref{fig:notations}, in which

\sstab(1) $\dddot{\mathcal{T}}[P_0$, $\ldots, P_{10}]$ is a trajectory having 11 data points,

\sstab (2) the set of two continuous line segments $\{\vv{P_0P_4}$, $\vv{P_4P_{10}}$\} (Left) and the set of four continuous line segments $\{\vv{P_0P_2}$, $\vv{P_2P_4}$, $\vv{P_4P_7}$, $\vv{P_7P_{10}}$\} (Right) are two piecewise line representations of trajectory $\dddot{\mathcal{T}}$,

\sstab(3) $ped(P_4, \vv{P_0P_{10}})=|\vv{P_4P^*_4}|$, where $P^*_4$ is the perpendicular point of $P_4$ \wrt line segment $\vv{P_0P_{10}}$, and

\sstab (4) $sed(P_4, \vv{P_0P_{10}})= |\vv{P_4P'_4}|$, $sed(P_2, \vv{P_0P_{4}})= |\vv{P_2P'_2}|$ and $sed(P_7, \vv{P_4P_{10}})$ $=$ $|\vv{P_7P'_7}|$,
where points $P'_4$, $P'_2$ and $P'_7$ are the synchronized points of $P_4$, $P_2$ and $P_7$ \wrt line segments $\vv{P_0P_{10}}$, $\vv{P_0P_{4}}$ and $\vv{P_4P_{10}}$, respectively. \eop
\end{example}

\stitle{Error bounded algorithms}. Given a trajectory \trajec{T} and its compression  algorithm $\mathcal{A}$ using \sed (respectively \ped) that produces another trajectory \trajec{T'},
we say that algorithm $\mathcal{A}$ is error bounded by $\epsilon$ if  for each point $P$ in \trajec{T}, there exist points $P_j$ and $P_{j+1}$ in \trajec{T'} such that $sed(P, \mathcal{L}(P_j,P_{j+1}))\le \epsilon$ (respectively $ped(P, \mathcal{L}(P_j,P_{j+1}))\le \epsilon$).

\subsection{The Optimal \lsa Algorithm}
\label{subsec-optimal}

Given a trajectory \trajec{T}${[P_0, \ldots, P_n]}$ and an error bound $\epsilon$, the optimal trajectory simplification problem, as formulated by Imai and Iri in \cite{Imai:Optimal}, can be solved in two steps: (1) construct a reachability graph $G$ of \trajec{T} and (2) search a shortest path from $P_0$ to $P_{n}$ in graph $G$.

The reachability graph of a trajectory \trajec{T}${[P_0, \ldots, P_n]}$ \wrt a bound $\epsilon$ is an unweighted graph $G(V, E)$, where (1) $V = \{P_0, \ldots, P_n\}$, and (2) for any nodes $P_s$ and $P_{s+k} \in V$ ($s\ge 0, k>0, s+k\le n$), edge $(P_s, P_{s+k}) \in E$ if and only if the distance of each point $P_{s+i}$ $(i\in[0,k])$ to line segment $\vv{P_sP_{s+k}}$ is no greater than $\epsilon$.

Observe that in the reachability graph $G$, (1) a path from nodes $P_0$ to $P_{n}$ is a representation of trajectory \trajec{T}. The path also reveals the subset of points of \trajec{T} used in the approximate trajectory, (2) the path length corresponds to the number of line segments in the approximation trajectory, and
(3) a shortest path is an optimal representation of trajectory \trajec{T}.

Constructing the reachability graph $G$ needs to check for all pair of points $P_s$ and $P_{s+k}$ whether the distances of all points ($P_{s+i}$, $0<i<k$) to the line segment $\vv{P_sP_{s+k}}$ are less than $\epsilon$.
There are $O(n^2)$ pairs of points in the trajectory and checking the error of all points $P_{s+i}$ to a line segment $\vv{P_sP_{s+k}}$ takes $O(n)$ time.
Thus, the construction step takes $O(n^3)$ time.
Finding shortest paths on unweighted graphs can be done in linear time. Hence, the brute-force algorithm takes $O(n^3)$ time in total.

Though the brute-force algorithm was initially developed using \ped, it can be used for \sed.
As pointed out in \cite{Chan:Optimal}, the construction of the reachability graph $G$ using \ped can be implemented in $O(n^2)$ time using the \textit{sector intersection} mechanism (see Section \ref{sub-ci-ped}). However, the \textit{sector intersection} mechanism cannot work with \sed.  Hence, the construction of the reachability graph $G$ using \sed remains in $O(n^3)$ time, and the brute-force algorithm using \sed remains in  $O(n^3)$ time.

\subsection{Sector Intersection based Algorithms using \ped}
\label{sub-ci-ped}

\begin{figure*}[tb!]
	\centering
	\includegraphics[scale=0.78]{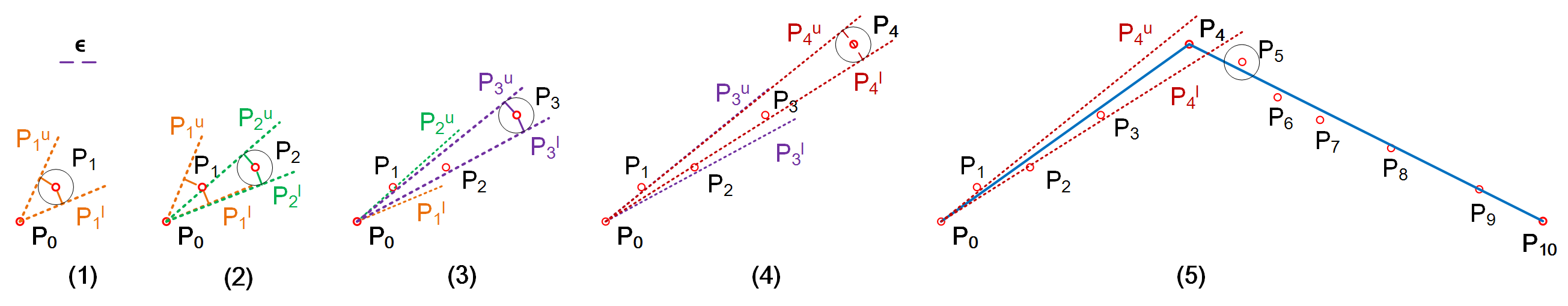}
	\caption{\small Trajectory $\dddot{\mathcal{T}}[P_0, \ldots, P_{10}]$ in Figure~\ref{fig:notations} is compressed into two line segments by the Sector Intersection algorithm  \cite{Williams:Longest, Sklansky:Cone}.}
	\vspace{-1ex}
	\label{fig:sleeve}
\end{figure*}

The sector intersection (\cia) algorithm \cite{Williams:Longest, Sklansky:Cone} was developed for graphic and pattern recognition in the late 1970s, for the approximation of arbitrary planar curves by linear segments or finding a polygonal approximation of a set of input data points in a 2D Cartesian coordinate system. The \sleeve algorithm \cite{Zhao:Sleeve} in the cartographic discipline essentially applies the same idea as the \cia algorithm.
Further, \cite{Dunham:Cone}  optimized algorithm \cia by considering the distance between a potential end point and the initial point of a line segment. It is worth pointing out that all these \cia based algorithms use the perpendicular Euclidean distances.


Given a sequence of data points $[P_{s}, P_{s+1}, \ldots, P_{s+k}]$ and an error bound $\epsilon$, the \cia based algorithms process the input points one by one in order, and produce a simplified polyline.  Instead of using the distance threshold $\epsilon$ directly, the \cia based algorithms convert the distance tolerance into a variable angle tolerance for testing the successive data points.

For the start data point $P_s$, any point $P_{s+i}$ and $|\vv{P_sP_{s+i}}|>\epsilon$ ($i\in[1, k]$), there are two directed lines $\vv{P_sP^u_{s+i}}$ and $\vv{P_sP^l_{s+i}}$ such that $ped(P_{s+i}, \vv{P_sP^u_{s+i}})$ $=$ $ped(P_{s+i}, \vv{P_sP^l_{s+i}}) = \epsilon$ and either ($\vv{P_sP^l_{s+i}}.\theta < \vv{P_sP^u_{s+i}}.\theta ~and~\vv{P_sP^u_{s+i}}.\theta - \vv{P_sP^l_{s+i}}.\theta <\pi$) or ($\vv{P_sP^l_{s+i}}.\theta > \vv{P_sP^u_{s+i}}.\theta ~and~ \vv{P_sP^u_{s+i}}.\theta - \vv{P_sP^l_{s+i}}.\theta < -\pi)$. Indeed, they forms a \emph{sector} \sector{(P_s, P_{s+i}, \epsilon)} that takes $P_s$ as the center point and $\vv{P_sP^u_{s+i}}$ and $\vv{P_sP^l_{s+i}}$ as the border lines.
Then there exists a data point $Q$ such that for any data point $P_{s+i}$ ($i \in [1, ... k]$), its perpendicular Euclidean distance to
directed line $\overline{P_sQ}$ is no greater than the error bound $\epsilon$ if and only if the $k$ sectors \sector{(P_s, P_{s+i}, \epsilon)} ($i\in[1,k]$) share common data points other than $P_s$, \ie $\bigsqcap_{i=1}^{k}$\sector{(P_s, P_{s+i}, \epsilon)} $\ne \{P_s\}$ \cite{Williams:Longest, Sklansky:Cone,Zhao:Sleeve}.

The point $Q$ may not belong to $\{P_{s}, P_{s+1},$ $\ldots, P_{s+k}\}$.
However, if $P_{s+i}$ ($1\le i\le k$) is chosen as $Q$, then
for any data point $P_{s+j}$ ($j \in [1, ... i]$), its perpendicular Euclidean distance to
line segment $\overline{P_sP_{s+i}}$ is no greater than the error bound $\epsilon$ if and only if $\bigsqcap_{j=1}^{i}$\sector{(P_s, P_{s+j}, \epsilon/2)} $\ne \{P_s\}$, as pointed out in \cite{Zhao:Sleeve}.

That is, {\em these \cia based algorithms can be easily adopted for trajectory compression using \ped although they have been overlooked by existing trajectory simplification studies}.
The \cia based algorithms run in $O(n)$ time and $O(1)$ space and are one-pass algorithms.

However, if we use \sed instead of \ped, then  ``the $k$ sectors \sector{(P_s, P_{s+i}, \epsilon)} ($i\in[1,k]$) sharing common data points other than $P_s$" cannot ensure ``there exists a data point $Q$ such that for any data point $P_{s+i}$ ($i \in [1, ... k]$), its synchronized Euclidean distance to directed line $\overline{P_sQ}$ is no greater than the error bound $\epsilon$''.
Hence, the \cia mechanism is \ped specific, and not applicable for \sed.

We next illustrate how the \cia based algorithms can be used for trajectory compression with an example.

\begin{example}
\label{exm-alg-sleeve}
Consider Figure~\ref{fig:sleeve}. An \cia based simplification algorithm takes as input a trajectory $\dddot{\mathcal{T}}[P_0, \ldots, P_{10}]$, and returns a simplified ployline consisting of two line segments $\vv{P_0P_4}$ and  $\vv{P_4P_{10}}$. Initially, point $P_0$ is the start point.

\sstab(1) Point $P_1$ is firstly read, and the sector \sector{(P_0, P_{1}, \epsilon/2)} of $P_1$ is created as shown in Figure~\ref{fig:sleeve}.(1).

\sstab(2) Then $P_2$ is read, and the sector \sector{(P_0, P_{2}, \epsilon/2)}  is created for $P_2$. The intersection of sectors \sector{(P_0, P_{1}, \epsilon/2)} and  \sector{(P_0, P_{2}, \epsilon/2)} contains data points other than $P_0$,  which has an up border line $P_0P_2^u$ and a low border line $P_0P_1^l$, as shown in Figure~\ref{fig:sleeve}.(2).

Similarly, points $P_3$ and $P_4$ are processed, as shown in Figures~\ref{fig:sleeve}.(3) and \ref{fig:sleeve}.(4), respectively.

\sstab(3) When point $P_5$ is read,  line segment $\vv{P_0P_4}$ is produced, and point $P_4$ becomes the start point, as $\bigsqcap_{i=1}^{4}$\sector{(P_0, P_{s+i}, \epsilon/2)} $\ne\{P_0\}$ and $\bigsqcap_{i=1}^{5}$\sector{(P_0, P_{s+i}, \epsilon/2)} $=\{P_0\}$ as shown in Figure~\ref{fig:sleeve}.(5).

\sstab(4) Points $P_5, \ldots, P_{10}$ are processed similarly one by one in order, and finally the algorithm outputs another line segment $\vv{P_4P_{10}}$ as shown in Figure~\ref{fig:sleeve}.(5). \eop
\end{example}

\subsection{Intersection Computation of Convex Polygons}
\label{subsec-cpi}

We also employ and revise a convex polygon intersection algorithm developed in~\cite{ORourke:Intersection}, whose basic idea is straightforward.
Assume \kwlog that the edges of polygons $\mathcal{R}_1$ and $\mathcal{R}_2$ are oriented counterclockwise, and $\vv{A} = (P_{s_A}, P_{e_A})$ and $\vv{B} = (P_{s_B}, P_{e_B})$ are two (directed) edges on $\mathcal{R}_2$ and $\mathcal{R}_1$, repectively (see Figure~\ref{fig:c-poly-inter}).

The algorithm has $\vv{A}$ and $\vv{B}$ ``chasing'' one another, \ie moves $\vv{A}$ on $\mathcal{R}_2$ and $\vv{B}$ on $\mathcal{R}_1$ counter-clockwise step by step under certain rules, so that they meet at every crossing of $\mathcal{R}_1$ and $\mathcal{R}_2$.
The rules, called \emph{advance rules}, are carefully designed depending on geometric conditions of $\vv{A}$ and $\vv{B}$.
Let $\vv{A} \times \vv{B}$ be the cross product of (vectors) $\vv{A}$ and $\vv{B}$, and $\mathcal{H}(\vv{A})$ be the open half-plane to the left of $\vv{A}$, the rules are as follows:

\sstab {\em Rule (1)}: If $\vv{A} \times \vv{B} < 0$ and $P_{e_A} \not \in \mathcal{H}(\vv{B})$, or $\vv{A} \times \vv{B} \ge 0$ and $P_{e_B} \in \mathcal{H}(\vv{A})$, then $\vv{A}$ is advanced a step.

For example, in Figure~\ref{fig:c-poly-inter}.(1) and~\ref{fig:c-poly-inter}.(2), $\vv{A}$ moves forward a step as  $\vv{A} \times \vv{B} > 0$ and $P_{e_B} \in \mathcal{H}(\vv{A})$.

\sstab {\em Rule (2)}: If $\vv{A} \times \vv{B} \ge 0$ and $P_{e_B} \not \in \mathcal{H}(\vv{A})$, or $\vv{A} \times \vv{B} < 0$ and $P_{e_A} \in \mathcal{H}(\vv{B})$, then  $\vv{B}$ is advanced a step.

For example, in Figure~\ref{fig:c-poly-inter}.(6) and~\ref{fig:c-poly-inter}.(7), $\vv{B}$ moves forward a step as $\vv{A} \times \vv{B} < 0$ and $P_{e_A} \in \mathcal{H}(\vv{B})$.

\stitle{Algorithm \cpia}. The complete algorithm is shown in Figure~\ref{alg:c-poly-inter}.
Given two convex polygons $\mathcal{R}_1$ and $\mathcal{R}_2$, algorithm \cpia first arbitrarily sets directed edge $\vv{A}$ on $\mathcal{R}_2$ and directed edge $\vv{B}$ on $\mathcal{R}_1$, respectively (line 1).
It then checks the intersection of edges $\vv{A}$ and $\vv{B}$. If $\vv{A}$ intersects $\vv{B}$ (line 3), then the algorithm checks for some special termination conditions (\eg if $\vv{A}$ and $\vv{B}$ are overlapped and, at the same time, polygons $\mathcal{R}_1$ and $\mathcal{R}_2$ are on the opposite sides of the overlapped edges, then the process is terminated) (line 4), and records the inner edge, which is a boundary segment of the intersection polygon (line 5).
After that, the algorithm moves on $\vv{A}$ or $\vv{B}$ one step under the advance rules (lines 6--11).
The above processes repeated, until both $\vv{A}$ and $\vv{B}$ cycle their polygons (line 12).
Next, the algorithm handles three special cases of the polygons $\mathcal{R}_1$ and $\mathcal{R}_2$, \ie $\mathcal{R}_1$ is inside of $\mathcal{R}_2$, $\mathcal{R}_2$ is inside of $\mathcal{R}_1$, and $\mathcal{R}_1 \bigsqcap \mathcal{R}_2 = \emptyset$ (line 13).
At last, it returns the intersection polygon (line 14).

The algorithm has a time complexity of $O(|\mathcal{R}_1| + |\mathcal{R}_2|)$, where $|\mathcal{R}|$ is the number of edges of polygon $\mathcal{R}$.
It is worth pointing out that $|\mathcal{R}_1 \bigsqcap \mathcal{R}_2| \le (|\mathcal{R}_1| + |\mathcal{R}_2|)$.

\begin{figure}[tb!]
	\begin{center}
		{\small
			\begin{minipage}{3.3in}
				\myhrule
				\vspace{-1ex}
				\mat{0ex}{
					{\bf Algorithm}~\cpia($\mathcal{R}_1$, $\mathcal{R}_2$) \\
					\bcc \hspace{1ex}\=  \SET $\vv{A}$ and $\vv{B}$ {arbitrarily} on $\mathcal{R}_2$ and $\mathcal{R}_1$, respectively;\\
					\icc \hspace{1ex}\= \Repeat \\
					\icc \>\hspace{2ex} \If $\vv{A} \bigsqcap \vv{B} \ne \emptyset$ \Then \\
					\icc \>\hspace{4ex} {Check for termination;} \\
					\icc \>\hspace{4ex} Update an inside flag;\\
					\icc \> \hspace{2ex} \If ($\vv{A}\times\vv{B} < 0$ \And $P_{e_A} \not \in \mathcal{H}(\vv{B})$) \Or \\
					\icc \> \hspace{4ex} ($\vv{A} \times \vv{B} \ge 0$ \And $P_{e_B} \in \mathcal{H}(\vv{A})$) \Then \\
					\icc \> \hspace{4ex} advance $\vv{A}$ one step; \\
					\icc \> \hspace{2ex} \ElseIf ($\vv{A} \times \vv{B} \ge 0$ \And $P_{e_B} \not \in \mathcal{H}(\vv{A})$) \Or\\
					\icc \> \hspace{4ex}  ($\vv{A} \times \vv{B} < 0$ \And $P_{e_A} \in \mathcal{H}(\vv{B})$) \Then \\
					\icc \> \hspace{4ex} advance $\vv{B}$ one step; \\
					\icc \hspace{0ex} \Until both $\vv{A}$ and $\vv{B}$ cycle their polygons \\
					\icc \hspace{0ex} \Handle $\mathcal{R}_1 \subset \mathcal{R}_2$ and $\mathcal{R}_2 \subset \mathcal{R}_1$ and $\mathcal{R}_1 \bigsqcap \mathcal{R}_2 = \emptyset$ cases; \\
					\icc \hspace{0ex} \Return $\mathcal{R}_1 \bigsqcap \mathcal{R}_2$.
				}
				\vspace{-2ex}
				\myhrule
			\end{minipage}
		}
	\end{center}
	\vspace{-2ex}
	\caption{\small Algorithm for convex polygons intersection~\cite{ORourke:Intersection}.}
	\label{alg:c-poly-inter}
	\vspace{-2ex}
\end{figure}

\begin{example}
Figure~\ref{fig:c-poly-inter} shows a running example of the convex polygon intersection algorithm \cpia.

\sstab(1) Initially, directed edges $\vv{A}$ and $\vv{B}$ are on polygons $\mathcal{R}_2$ and $\mathcal{R}_1$, respectively, such that $\vv{A} \bigsqcap \vv{B} = \{P_1\}$, \ie $\vv{A}$ and $\vv{B}$ intersect on point $P_1$, as shown in Figure~\ref{fig:c-poly-inter}.(1).

\sstab(2) Then, by advance rule (1), edge $\vv{A}$ moves on a step and makes $\vv{A} \bigsqcap \vv{B} = \emptyset$ as shown in Figure~\ref{fig:c-poly-inter}.(2).
After 7 steps of moving of edge $\vv{A}$ or $\vv{B}$, each by an advance rule, $\vv{A}$ and $\vv{B}$ intersect on $P_2$, as shown in Figure~\ref{fig:c-poly-inter}.(6).

\sstab(3) Next, edge $\vv{B}$ moves on a step, and makes $\vv{A} \bigsqcap \vv{B} = \emptyset$, as shown in Figure~\ref{fig:c-poly-inter}.(7).

\sstab(4) After 6 steps of moving of edge $\vv{B}$ or $\vv{A}$ one by one, both edges $\vv{A}$ and $\vv{B}$ have finished their cycles as shown in Figure~\ref{fig:c-poly-inter}.(8).

\sstab(5) The algorithm finally returns the intersection polygon as shown in Figure~\ref{fig:c-poly-inter}.(9). \eop
\end{example}

\begin{figure}[tb!]
	\centering
	\includegraphics[scale=0.92]{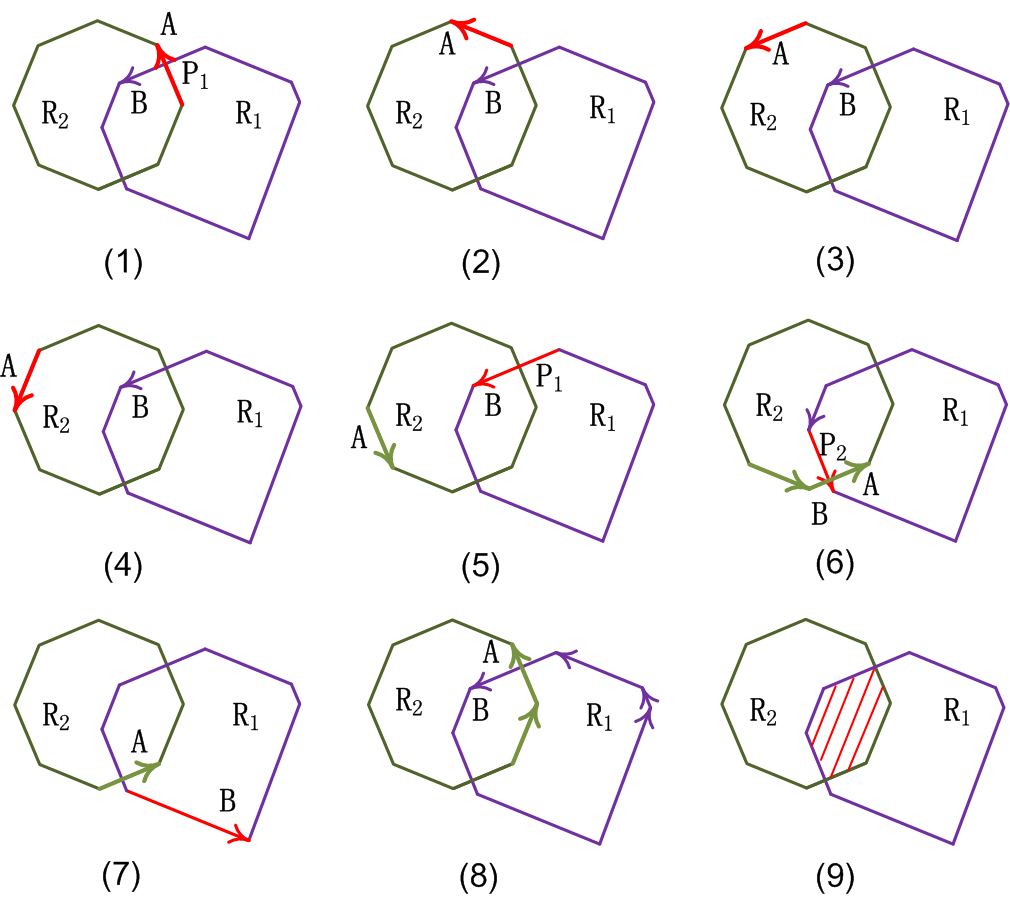}
	\caption{\small A running example of convex polygons intersection.}
	\vspace{-1ex}
	\label{fig:c-poly-inter}
\end{figure}


\section{Local Synchronous Distance Checking}
\label{sec-localcheck}


In this section, we develop a local synchronous distance checking approach such that each point in a trajectory is checked only once in $O(1)$ time during the entire process of trajectory simplification, by substantially extending the \textit{sector intersection} method in Section~\ref{sub-ci-ped} from a 2D space to a Spatio-Temporal 3D space, which lays down the key for the one-pass trajectory simplification algorithms using \sed (Section~\ref{sec-alg}).

We consider a sub-trajectory $\dddot{\mathcal{T}}_s[P_s, \ldots, P_{s+k}]$, an error bound $\epsilon$, and a 3D Cartesian coordinate system whose origin, $x$-axis, $y$-axis and $t$-axis  are $P_s$, longitude, latitude and time, respectively.

\subsection{Spatio-Temporal Cone Intersection}

We first present the \textit{spatio-temporal cone intersection} method in a 3D Cartesian coordinate system, which extends the \textit{sector intersection} method~\cite{Williams:Longest, Sklansky:Cone, Zhao:Sleeve}. 


\stitle{Synchronous Circles (\circle{})}. The synchronous circle of a data point $P_{s+i}$ ($1\le i\le k$) in $\dddot{\mathcal{T}}_s$ \wrt an error bound $\epsilon$, denoted as \circle{(P_{s+i}, \epsilon)}, or \circle{_{s+i}} in short, is a circle on the plane $P.t-P_{s+i}.t = 0$ such that $P_{s+i}$ is its center and $\epsilon$ is its radius.

Figure~\ref{fig:cis} shows two synchronous circles, \circle{(P_{s+i}, \epsilon)} of point $P_{s+i}$ and \circle{(P_{s+k}, \epsilon)} of point $P_{s+k}$.
It is easy to know that for any point in the area of a circle \circle{(P_{s+i}, \epsilon)}, its distance to $P_{s+i}$ is no greater than $\epsilon$.


\stitle{Spatio-temporal cones (\cone{})}. The spatio-temporal cone (or simply \textit{cone}) of a data point $P_{s+i}$ ($1\le i\le k$) in $\dddot{\mathcal{T}}_s$ \wrt a point $P_s$ and an error bound $\epsilon$, denoted as \cone{(P_s, \mathcal{O}(P_{s+i}, \epsilon))}, or \cone{_{s+i}} in short, is an oblique circular cone such that point $P_s$ is its apex and the synchronous circle $\mathcal{O}(P_{s+i}, \epsilon)$ of point $P_{s+i}$ is its base.

Figure~\ref{fig:cis} also illustrates two example spatio-temporal cones: \cone{(P_s, \mathcal{O}(P_{s+i}, \epsilon))} {(purple)} and \cone{(P_s, \mathcal{O}(P_{s+k}, \epsilon))} (red), with the same apex $P_s$ and error bound $\epsilon$.

\begin{figure}[tb!]
	\centering
	\includegraphics[scale=0.66]{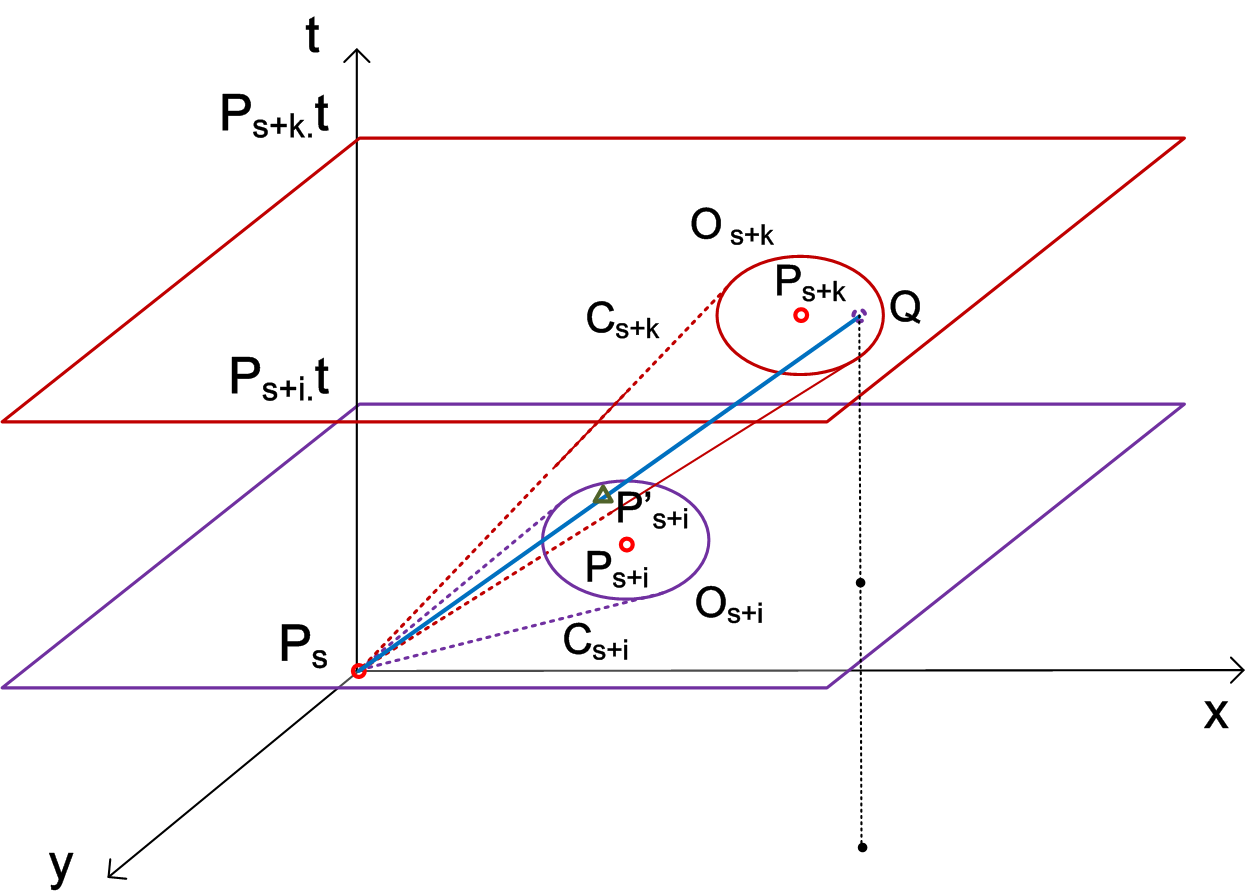}
	\caption{\small Examples of spatio-temporal cones in a 3D Cartesian coordinate system, where (1) $P_s$, $P_{s+i}$ and $P_{s+k}$ are three points, (2) \circle{_{s+i}} and \circle{_{s+k}} are two synchronous circles, (3) \cone{_{s+i}} and \cone{_{s+k}} are two spatio-temporal cones, (4) $Q$ is a point in synchronous circle \circle{_{s+k}}, and (5) $P'_{s+i}$ is the intersection point of line $\overline{P_sQ}$ and synchronous circle \circle{_{s+i}}.}
	\label{fig:cis}
\end{figure}



\begin{prop}
\label{prop-3d-syn-point}
Given a sub-trajectory $[P_s, \ldots, P_{s+k}]$ and a point $Q$ in the area of synchronous circle \circle{(P_{s+k}, \epsilon)}, the intersection point $P'_{s+i}$ of the directed line segment $\vv{P_sQ}$ and the plane $P.t - P_{s+i}.t = 0$ is the synchronized point of $P_{s+i}$ ($1\le i\le k$) \wrt  $\vv{P_sQ}$, and the distance $|\vv{P_{s+i}P'_{s+i}}|$ from $P_{s+i}$ to $P'_{s+i}$ is the synchronous distance of $P_{s+i}$ to $\vv{P_sQ}$.
\end{prop}

\begin{proof}\ It suffices to show that $P'_{s+i}$ is indeed a synchronized point $P_{s+i}$ \wrt $\vv{P_sQ}$.
The intersection point $P'_{s+i}$ satisfies that $P'_{s+i}.t = P_{s+i}.t$ and
$\frac{P'_{s+i}.t - P_{s}.t}{Q.t - P_{s}.t}$ = $\frac{P_{s+i}.t - P_{s}.t}{Q.t - P_{s}.t}$  =
$\frac{|\vv{P_sP'_{s+i}}|}{|\vv{P_sQ}|}$ =
$\frac{P'_{s+i}.x - P_{s}.x}{Q.x - P_{s}.x}$ = $\frac{P'_{s+i}.y - P_{s}.y}{Q.y - P_{s}.y}$.
Hence, by the definition of synchronized points, we have the conclusion. \eop
\end{proof}

\begin{prop}
\label{prop-3d-ci}
Given a sub-trajectory $[P_s,...,P_{s+k}]$ and an error bound $\epsilon$, there exists a point $Q$ such that $Q.t = P_{s+k}.t$ and $sed(P_{s+i}, \vv{P_sQ})\le \epsilon$ for each $i \in [1,k]$ if and only if $\bigsqcap_{i=1}^{k}$\cone{(P_s, \mathcal{O}(P_{s+i}, \epsilon))} $\ne \{P_s\}$.
\end{prop}

\begin{proof}\
Let $P'_{s+i}$ ($i\in[1, k]$) be the intersection point of line segment $\vv{P_sQ}$ and the plane $P.t - P_{s+i}.t$ = $0$.
By Proposition~\ref{prop-3d-syn-point}, $P'_{s+i}$ is the synchronized point of $P_{s+i}$ \wrt $\vv{P_sQ}$.

Assume first that $\bigsqcap_{i=1}^{k}$\cone{(P_s, \mathcal{O}(P_{s+i}, \epsilon))} $\ne \{P_s\}$. Then there must exist a point $Q $ in the area of the  synchronous circle \circle{(P_{s+k}, \epsilon)} such that $\vv{P_sQ}$ passes through all the cones \cone{(P_s, \mathcal{O}(P_{s+i}, \epsilon))} $i\in[1, k]$. Hence,  $Q.t = P_{s+k}.t$.
We also have $sed(P_{s+i}, \vv{P_sQ}) = |\vv{P'_{s+i}P_{s+i}}| \le \epsilon$ for each $i \in [1, k]$  since $P'_{s+i}$  is in the area of circle  \circle{(P_{s+i}, \epsilon)}.

Conversely, assume that there exists a point $Q$ such that $Q.t = P_{s+k}.t$ and $sed(P_{s+i}, \vv{P_sQ})\le\epsilon$ for all $P_{s+i}$ ($i \in [1,k]$). Then $|\vv{P'_{s+i}P_{s+i}}| \le \epsilon$ for all $i \in [1, k]$. Hence, we have  $\bigsqcap_{i=1}^{k}$\cone{(P_s, \mathcal{O}(P_{s+i}, \epsilon))} $\ne \{P_s\}$. \eop
\end{proof}

By Proposition~\ref{prop-3d-ci}, we now have a spatio-temporal cone intersection method in a 3D Cartesian coordinate system, which significantly extends the sector intersection method~\cite{Williams:Longest, Sklansky:Cone, Zhao:Sleeve} {from a 2D space to a Spatio-Temporal 3D space}.

\subsection{Circle Intersection}
\label{subsec-ProjectionCircle}

For spatio-temporal cones with the same apex $P_s$, the checking of their intersection can be computed by a much simpler way, \ie the checking of intersection of cone projection circles on a plane, as follows.

\stitle{Cone projection circles}. The projection of a cone \cone{(P_s, \mathcal{O}(P_{s+i}, \epsilon))} on a plane $P.t- t_c = 0$ ($t_c > P_s.t$) is a circle \pcircle{(P^c_{s+i}, r^c_{s+i})}, or \pcircle{_{s+i}} in short, such that
(1) $P^c_{s+i}.x = P_s.x +  c\cdot(P_{s+i}.x- P_{s}.x)$,
(2) $P^c_{s+i}.y = P_s.y +  c\cdot(P_{s+i}.y- P_{s}.y)$,
(3) $P^c_{s+i}.t = t_c$ and
(4) $r^c_{s+i} =c\cdot\epsilon$, where $c=\frac{t_c - P_s.t}{P_{s+i}.t - P_s.t}$ .

\vspace{.5ex}

In Figure~\ref{fig:pcircle}, the green dashed circles \pcircle{(P^c_{s+i}, r^c_{s+i})} and \pcircle{(P^c_{s+k}, r^c_{s+k})} on plane ``$P.t-t_{c}=0$" are the projection circles of cones \cone{(P_s, \mathcal{O}(P_{s+i}, \epsilon))} and \cone{(P_s, \mathcal{O}(P_{s+k}, \epsilon))} on the plane.

\begin{prop}
\label{prop-circle-intersection}
Given a sub-trajectory $[P_s,\ldots, P_{s+k}]$, an error bound $\epsilon$, and any $t_c > P_s.t$, there exists a point $Q$ such that $Q.t = P_{s+k}.t$ and $sed(P_{s+i}, \vv{P_sQ})\le \epsilon$ for all points $P_{s+i}$ ($i \in [1,k]$) if and only if $\bigsqcap_{i=1}^{k}$\pcircle{(P^c_{s+i}, r^c_{s+i})} $\ne \emptyset$.
\end{prop}

\begin{proof}\
By Proposition~\ref{prop-3d-ci}, it suffices to show that $\bigsqcap_{i=1}^{k}$ \pcircle{(P^c_{s+i}, r^c_{s+i})} $\ne \emptyset$ if and only if $\bigsqcap_{i=1}^{k}$\cone{(P_s, \mathcal{O}(P_{s+i}, \epsilon))}$\ne \{P_s\}$, which is obvious. Hence, we have the conclusion. \eop
\end{proof}


Proposition~\ref{prop-circle-intersection} tells us that the intersection checking of spatio-temporal cones can be reduced to simply check the intersection of cone projection circles on a plane. 

\begin{figure}[tb!]
	\centering
	\includegraphics[scale=0.7]{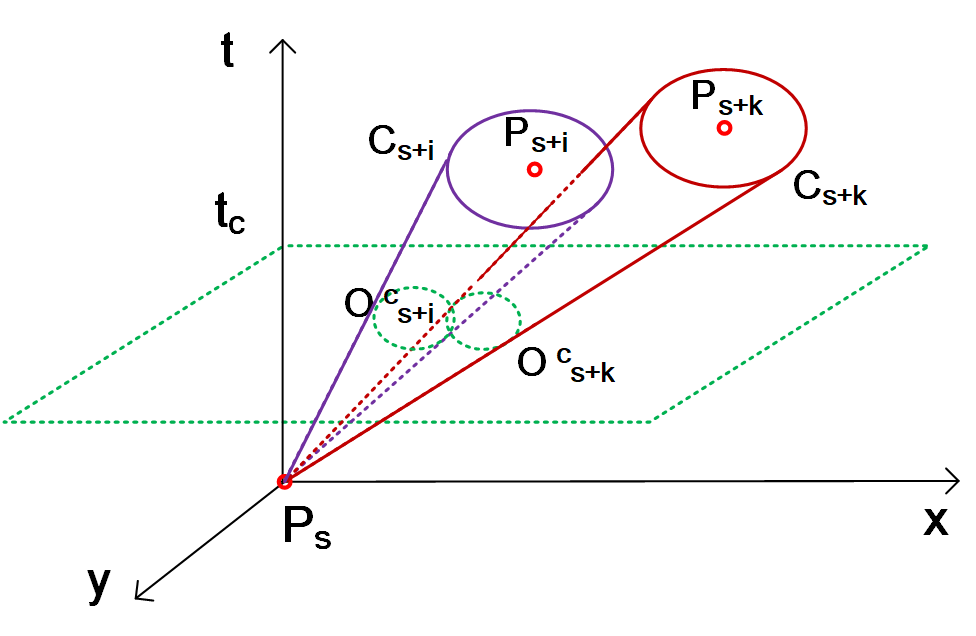}
	\caption{\small Cone projection circles.}
	\vspace{-1ex}
	\label{fig:pcircle}
\end{figure}

\subsection{Inscribed Regular Polygon Intersection}
\label{subsec-RPI}

Finding the common intersection of $n$ circles on a plane has a time complexity of ${O(n\log n)}$~\cite{Shamos:Circle}, which cannot be used for designing one-pass trajectory simplification algorithms using \sed.
However, we can approximate a circle with its $m$-edge inscribed regular polygon, whose intersection can be computed more efficiently.

\stitle{Inscribed regular polygons ($\mathcal{R}$)}.
Given a cone projection circle \pcircle{(P, r)}, its inscribed $m$-edge regular polygon is denoted as $\mathcal{R}(V, E)$,
where (1) $V=\{v_1, \ldots, v_{m}\}$ is the set of vertexes that are defined by a polar coordinate system, whose origin is the center $P$ of \pcircle{}, as follows:

\vspace{-2ex}
\begin{equation*}
\label{equ-regular-polygon}
    \begin{aligned}
        \hspace{5ex}  v_j = (r, \frac{(j-1)}{m}2\pi), ~j \in [1, m], \\
    \end{aligned}
\end{equation*}
\ni and (2) $E= \{\vv{v_mv_1}\} \bigcup \{\vv{v_jv_{j+1}}\ |\ j\in [1, m-1]\}$ is the set of edges that are labeled with the subscript of their start points.


Figure~\ref{fig:polygons}.(1) illustrates the inscribed regular octagon ($m=8$) of a cone projection circle \pcircle{(P, r)}.

Let $\mathcal{R}_{s+i}$ ($1\le i \le k$) be the inscribed regular polygon of the cone projection  circle \pcircle{(P^c_{s+i}, r^c_{s+i})},
$\mathcal{R}^*_l$ ($1\le l\le k$) be the intersection $\bigsqcap_{i=1}^{l}\mathcal{R}_{s+i}$,
and $E^j$ ($1\le j \le m$) be the group of $k$ edges labeled with $j$ in all $\mathcal{R}_{s+i}$ ($i\in[1, k]$).
It is easy to verify that all edges in the same edge groups $E^j$ ($1\le j\le m$) are in parallel (or overlapping) with each other by the above definition of inscribed regular polygons, as illustrated in Figure~\ref{fig:polygons}.(2).

\begin{prop}
\label{prop-rp-intersection}
The intersection $\mathcal{R}^*_{l} \bigsqcap \mathcal{R}_{s+l+1}$ ($ 1\le l< k$) has at most $m$ edges, \ie at most one from each edge group.
\end{prop}

\eat{
\begin{theorem}
\label{prop-rp-intersection}
If $\mathcal{R}_i$, $i \in [1, k]$, are M-edges regular polygons on a plane which are built by equation (3), then the intersection polygon
$\mathcal{R}^*_k$ of all $\mathcal{R}_i$ includes at most one edge from an edge group, \eg the $j^{th}$ edge group.
\end{theorem}
}

\begin{proof}\
We shall prove this by contradiction.
Assume that $\mathcal{R}^*_{l} \bigsqcap \mathcal{R}_{s+l+1}$ has two distinct edges $\vv{A_i}$ and $\vv{A_{i'}}$  with the same label $j$ $(1\le j \le m)$, originally from
$\mathcal{R}_{s+i}$ and $\mathcal{R}_{s+i'}$  ($1\le i< i' \le l+1$).
Note that here $\mathcal{R}_{s+i} \bigsqcap \mathcal{R}_{s+i'} \ne \emptyset$ since $\mathcal{R}^*_l \bigsqcap \mathcal{R}_{s+l+1} \ne \emptyset$.
However, when $\mathcal{R}_{s+i} \bigsqcap \mathcal{R}_{s+i'} \ne \emptyset$, the intersection $\mathcal{R}_{s+i} \bigsqcap \mathcal{R}_{s+i'}$ cannot have
both edge $\vv{A_i}$ and edge $\vv{A_{i'}}$, which contradicts the assumption. \eop
\end{proof}

\begin{figure}[tb!]
	\centering
	\includegraphics[scale=0.85]{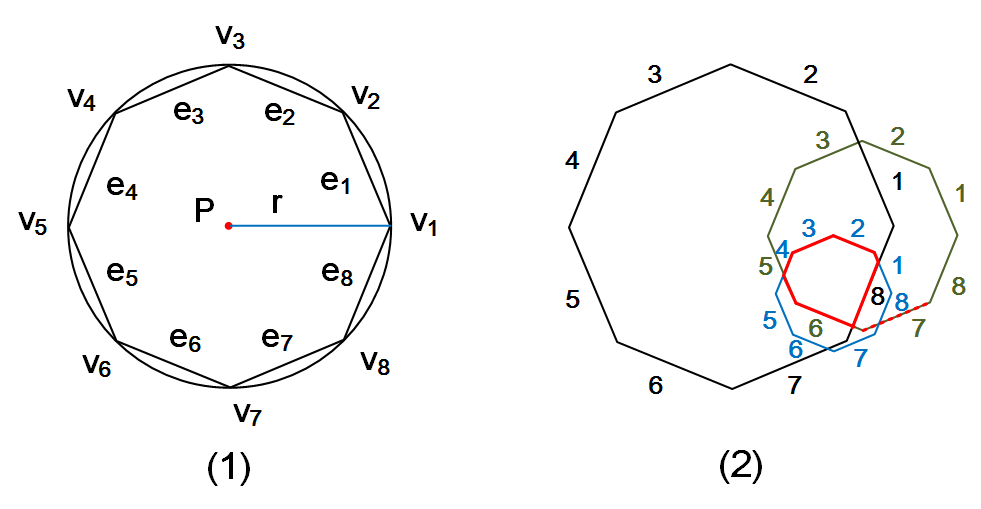}
	\vspace{0ex}
	\caption{\small Regular octagons and their intersections ($m =8$).}
	\vspace{-1ex}
	\label{fig:polygons}
\end{figure}

Figure~\ref{fig:polygons}.(2) shows the intersection polygon (red lines) of $\mathcal{R}_1$, $\mathcal{R}_2$ and $\mathcal{R}_3$ with $7$ edges, and here edges labeled with $7$ have no contributions to the resulting intersection polygon.

\begin{prop}
\label{prop-cpi-time}
The intersection of $\mathcal{R}^*_l$ and $\mathcal{R}_{s+l+1}$ ($ 1\le l< k$) can be done in $O(1)$
time.
\end{prop}

\begin{proof}\
The inscribed regular polygon $\mathcal{R}_{s+l+1}$ has $m$ edges, and intersection polygon $\mathcal{R}^*_l$ has at most $m$ edges by Proposition~\ref{prop-rp-intersection}.
As the intersection of two $m$-edge convex polygons can be computed in $O(m)$ time~\cite{ORourke:Intersection}, the intersection of polygons $\mathcal{R}^*_l$ and $\mathcal{R}_{s+l+1}$ can be done in $O(1)$ time for a fixed $m$. \eop
\end{proof}


\subsection{Speedup Inscribed Regular Polygon Intersection}
\label{subsec-fastRPI}

Observe that algorithm \cpia in Figure~\ref{alg:c-poly-inter} is for general convex polygons,
while the inscribed regular polygons $\mathcal{R}_{s+i}$ ($i\in[1, k]$) of the cone projection circles are constructed in a unified way,
which allows us to develop a fast method to compute their intersection.

Let $\vv{A} = (P_{s_A}, P_{e_A})$ and $\vv{B} = (P_{s_B}, P_{e_B})$  be two directed edges on polygons $\mathcal{R}_{s+l+1}$ and $\mathcal{R}^*_{l}$, respectively.
Again edges $\vv{A}$ and $\vv{B}$ are moved counter-clockwise. Note that $\vv{A}$ and $\vv{B}$ are advanced step by step each time by the two advancing rules of algorithm \cpia.
However, it is possible to advance $\vv{A}$ or $\vv{B}$ multiple steps each time.
For example, in Figure~\ref{fig:c-poly-inter}.(1)--(5), edge $\vv{A}$ successively moves four steps, each under the advance rule (1) ``($\vv{A} \times \vv{B} < 0$ and $P_{e_A} \not \in \mathcal{H}(\vv{B})$) or ($\vv{A} \times \vv{B} \ge 0$ and $P_{e_B} \in \mathcal{H}(\vv{A})$)'' of algorithm \cpia.
Alternatively, we can directly move $A$ from Figure~\ref{fig:c-poly-inter}.(1) to Figure~\ref{fig:c-poly-inter}.(5), by reducing four steps to one step only.

\begin{prop}
\label{prop-rule1}
If either $(\vv{A} \bigsqcap \vv{B} \ne \emptyset$ \And $\vv{A} \times \vv{B} < 0$ \And $P_{e_A} \not \in \mathcal{H}(\vv{B}))$ or $(\vv{A} \bigsqcap \vv{B} \ne \emptyset$ \And $\vv{A} \times \vv{B} \ge 0$ \And $P_{e_B} \in \mathcal{H}(\vv{A}))$ holds, then $\vv{A}$ advances $s$ steps such that

\vspace{-1ex}
\begin{equation*}
\label{equ-rule1}
\small
    \hspace{2ex} s =  \left\{
    \begin{aligned}
        & 2\times(g(\vv{B}) - g(\vv{A}))  \hspace{5ex}~~if  ~g(\vv{B}) > g(\vv{A}) \\
        & {1}              \hspace{21ex}~if  ~g(\vv{A}) = g(\vv{B}) \\
        & 2\times(m+g(\vv{B}) - g(\vv{A})) ~~if  ~g(\vv{B}) < g(\vv{A}), \\
    \end{aligned}
    \right.       \hspace{6ex}{}
\end{equation*}
in which $g(e)$ denotes the label of edge $e$.
\end{prop}

\begin{proof}\ 
We first explain how the edge $\vv{A}$ advances.
Indeed, $\vv{A}$ is moved from its original position to its symmetric edge on $\mathcal{R}_{s+l+1}$ \wrt the symmetric line that is perpendicular to $\vv{B}$  on $\mathcal{R}^*_{l}$.
For example, in Figure~\ref{fig:r-poly-rule1}.(1), there is $\vv{A} \bigsqcap \vv{B} \ne \emptyset$ \And $\vv{A} \times \vv{B} \ge 0$ \And $P_{e_B} \in \mathcal{H}(\vv{A})$, hence $\vv{A}$ moves on. As $g(\vv{B})=3 > 1=g(\vv{A})$, $\vv{A}$ moves forward $2\times(g(\vv{B}) - g(\vv{A}))$ = $2\times(3-1)= 4$ steps.
Here, the label of edge $\vv{A}$ is changed to $5$, its symmetric edge $1$ on $\mathcal{R}_{s+l+1}$ \wrt the symmetric line that is perpendicular to $\vv{B}$ labeled with $3$  on $\mathcal{R}^*_{l}$.

We then present the proof.
If ($\vv{A} \bigsqcap \vv{B} \ne \emptyset$ \And $\vv{A} \times \vv{B} < 0$ \And $P_{e_A} \not \in \mathcal{H}(\vv{B})$) or ($\vv{A} \bigsqcap \vv{B} \ne \emptyset$ \And $\vv{A} \times \vv{B} \ge 0$ \And $P_{e_B} \in \mathcal{H}(\vv{A})$), then as all edges in the same edge groups $E^j$ ($1\le j\le m$) are in parallel with each other and by the geometric properties of regular polygon $\mathcal{R}_{s+k+1}$, it is easy to find that, for each position of $\vv{A}$ between its original to its opposite positions, we have (1) $\vv{A} \bigsqcap \vv{B} = \emptyset$, and (2) either $P_{e_A} \not \in \mathcal{H}(\vv{B})$ or $P_{e_B} \in \mathcal{H}(\vv{A})$. Hence, by the advance rule (1) of algorithm \cpia in Section~\ref{subsec-cpi}, edge $\vv{A}$ is always moved forward until it reaches the opposite position of its original one. From this, we have the conclusion. \eop
\end{proof}

\begin{prop}
\label{prop-rule2}
If either ($\vv{A} \bigsqcap \vv{B} \ne \emptyset$ \And $\vv{A} \times \vv{B} \ge 0$ \And $P_{e_B} \not \in \mathcal{H}(\vv{A})$) or ($\vv{A} \bigsqcap \vv{B} \ne \emptyset$ \And $\vv{A} \times \vv{B} < 0$ \And $P_{e_A} \in \mathcal{H}(\vv{B})$) holds, then edge $\vv{B}$ is directly moved to the edge after the one having the same edge group as edge $\vv{A}$.
\end{prop}

\begin{proof}\ 
We first explain how the edge $\vv{B}$ is moved forward.
For example, in Figure~\ref{fig:r-poly-rule1}.(2), $\vv{A} \bigsqcap \vv{B} \ne \emptyset$ \And $\vv{A} \times \vv{B} < 0$ \And $P_{e_A} \in \mathcal{H}(\vv{B})$, hence $\vv{B}$ is moved forward. As the edge $\vv{A}$ is labeled with 7,
$\vv{B}$ moves to the edge labeled with 8 on $\mathcal{R}^*_{l}$, which is the next of the edge labeled with 7 on $\mathcal{R}^*_{l}$.
Note that if the edge labeled with 8 were not actually existing in the intersection polygon $\mathcal{R}^*_{l}$, then $\vv{B}$ should repeatedly move on until it reaches the first ``real" edge on $\mathcal{R}^*_{l}$.

We then present the proof.
If ($\vv{A} \bigsqcap \vv{B} \ne \emptyset$ \And $\vv{A} \times \vv{B} \ge 0$ \And $P_{e_B} \not \in \mathcal{H}(\vv{A})$) or ($\vv{A} \bigsqcap \vv{B} \ne \emptyset$ \And $\vv{A} \times \vv{B} < 0$ \And $P_{e_A} \in \mathcal{H}(\vv{B})$), then it is also easy to find that, for each position of $\vv{B}$ between its original to its target positions (\ie the edge after the one having the same edge group as $\vv{A}$), we have (1) $\vv{A} \bigsqcap \vv{B} = \emptyset$, and (2) either $P_{e_B} \not \in \mathcal{H}(\vv{A})$ or $P_{e_A} \in \mathcal{H}(\vv{B})$. Hence, by the advance rule (2) of algorithm \cpia in Section~\ref{subsec-cpi}, edge $\vv{B}$ is always moved forward until it reaches the target position. From this, we have the conclusion. \eop
\end{proof}

\begin{figure}[tb!]
	\centering
	\includegraphics[scale=0.82]{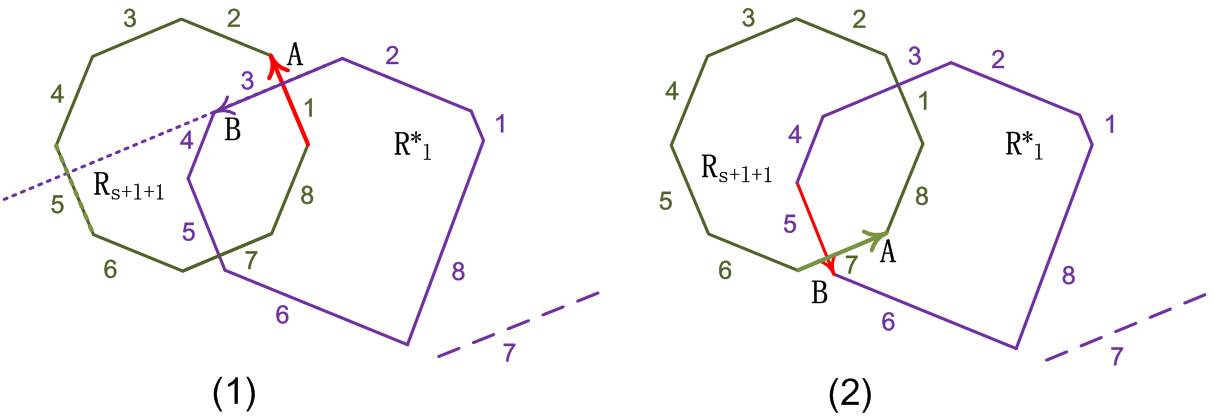}
	\vspace{-1ex}
	\caption{\small Examples of fast advancing rules.}
	\vspace{-1ex}
	\label{fig:r-poly-rule1}
\end{figure}


\eat{
\vspace{1ex}
\ni \emph{\underline{Rule 3}:
If $A \bigcap B = \emptyset$ and advances $A$, then moves $A$ and $B$ forward ``$a$ and $b$" steps respectively, where}
\begin{equation*}
\label{equ-rule3}
\small
    \hspace{0ex} (a,b) =  \left\{
    \begin{aligned}
        & (2,0), ~if~ g(next(A)) = g(B) ~and~ {outside}(A)  \\
        & (1,1), ~if~ g(next(A)) = g(B) ~and~ {inside}(A)\\
        & (1,0), ~else
    \end{aligned}
    \right.       \hspace{2ex}(5)
\end{equation*}
\emph{and procedure $inside()$ or $outside()$ is a checking of the ``{\emph{inside}}" flag (line 5) of the \cpia algorithm (Figure~\ref{alg:c-poly-inter}).}
\vspace{1ex}

For example, in Figure~\ref{fig:r-poly-inter}-(4), $A \bigcap B = \emptyset$ and advances $A$, hence, rule 3 is applied. Because $A$ is inside and $g(next(A))=3 = g(B)$, $A$ and $B$ both move forward one step (Figure~\ref{fig:r-poly-inter}-(5)).

\vspace{1ex}
\ni \emph{\underline{Rule 4}:
If $A \bigcap B = \phi$ and advances $B$, then moves $A$ and $B$ forward ``$a$ and $b$" steps respectively, where}
\begin{equation*}
\label{equ-rule4}
\small
    \hspace{0ex} (a,b) =  \left\{
    \begin{aligned}
        & (0,2), ~if~ g(next(B)) = g(A) ~and~ outside(B)\\
        & (1,1), ~if~ g(next(B)) = g(A) ~and~ inside(B)\\
        & (0,1), ~else
    \end{aligned}
    \right.       \hspace{2ex}(6)
\end{equation*}
\vspace{-1ex}

}

\stitle{Algorithm \rpia}.
The presented regular polygon intersection algorithm, \ie\ \rpia, is the optimized version  of the convex polygon intersection algorithm \cpia, by Propositions \ref{prop-rule1} and \ref{prop-rule2}. We also save vertexes of a polygon in a fixed size array, which is different from \cpia  that saves polygons in linked lists.
Considering the regular polygons each having a fixed number of vertexes/edges, marked from $1$ to $m$, this policy allows us to quickly address an edge or vertex by its label.

Given intersection polygon $\mathcal{R}^*_{l}$ of the preview $l$ polygons and the next approximate polygon $\mathcal{R}_{s+l+1}$, the algorithm \rpia returns $\mathcal{R}^*_{l+1}=\mathcal{R}^*_{l}  \bigsqcap \mathcal{R}_{s+l+1}$.
It runs the similar routine as the \cpia algorithm, except that (1) it saves polygons in arrays, and (2) the advance strategies are partitioned into two parts, \ie $\vv{A} \bigsqcap \vv{B} \ne \emptyset$ and $\vv{A} \bigsqcap \vv{B} = \emptyset$, where the former applies Propositions \ref{prop-rule1} and \ref{prop-rule2}, and the later remains the same as algorithm \cpia.


\vspace{0.5ex}
\stitle{Correctness and complexity analyses.}
Observe that algorithm \rpia basically has the same routine as algorithm \cpia, except that it fastens the advancing speed of directed edges $\vv{A}$ and $\vv{B}$ under certain circumstances as shown by Propositions \ref{prop-rule1} and \ref{prop-rule2}, which together ensure the correctness of \rpia. Moreover, algorithm \rpia runs in $O(1)$ time
by Proposition~\ref{prop-cpi-time}.


\section{One-Pass Trajectory Simplification}
\label{sec-alg}

\begin{figure*}[tb!]
	\centering
	\includegraphics[scale=0.8]{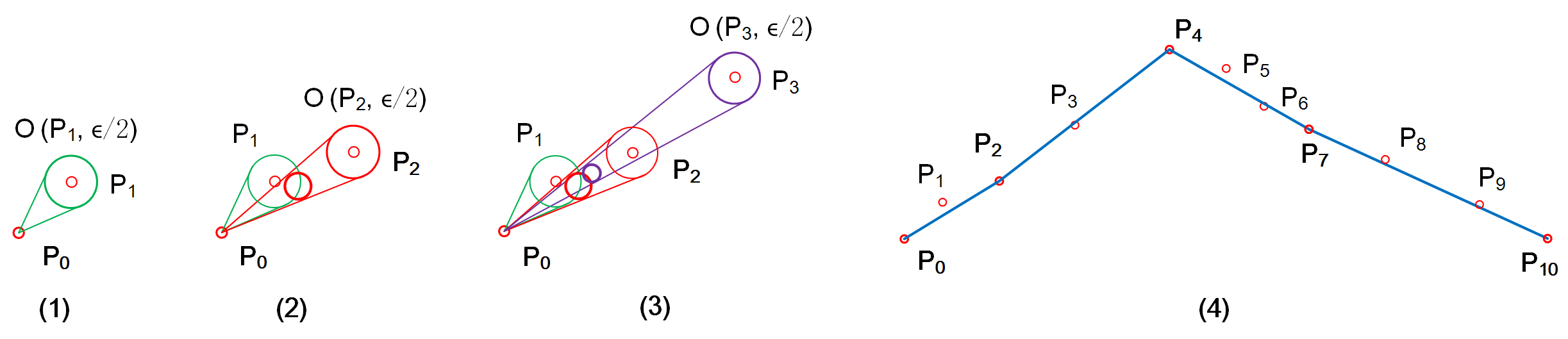}
	\caption{\small A running example of the \cist algorithm. The points and the oblique circular cones are projected on an x-y space. The trajectory $\dddot{\mathcal{T}}[P_0, \ldots, P_{10}]$ is compressed into four line segments.}
	\label{fig:exm-const}
\end{figure*}

Following \cite{Trajcevski:DDR,Lin:Operb}, we consider two classes of trajectory simplification.
The first one, referred to as \emph{strong simplification}, that takes as input a trajectory \trajec{T}, an error bound $\epsilon$ and the number $m$ of edges for inscribed regular polygons, and produces a simplified trajectory \trajec{T'} such that all data points in \trajec{T'} belong to \trajec{T}.
The second one, referred to as \emph{weak simplification}, that takes as input a trajectory \trajec{T}, an error bound $\epsilon$ and the number $m$ of edges for inscribed regular polygons, and produces a simplified trajectory \trajec{T'} such that some data points in \trajec{T'} may not belong to \trajec{T}. That is, weak simplification allows data interpolation.

The main result here is stated as follows.

\begin{theorem}
\label{prop-cist-op}
There exist one-pass, error bounded and strong and weak trajectory simplification algorithms using the synchronous Euclidean distance (\sed).
\end{theorem}

We shall prove this by providing such algorithms for both strong and weak trajectory simplifications, by employing the constant time synchronous distance checking technique developed in Section~\ref{sec-localcheck}.

\subsection{Strong Trajectory Simplification}

Recall that in Propositions~\ref{prop-3d-ci} and~\ref{prop-circle-intersection}, the point $Q$ may not be in the input sub-trajectory $[P_s,...,P_{s+k}]$.
If we restrict $Q=P_{s+k}$, the end point of the sub-trajectory, then the narrow cones whose base circles with a radius of $\epsilon/2$ suffice.

\begin{prop}
\label{prop-3d-ci-half}
Given a sub-trajectory $[P_s,...,P_{s+k}]$ and an error bound $\epsilon$, $sed(P_{s+i}, \vv{P_sP_{s+k}})\le \epsilon$ for each $i \in [1,k]$ if  $\bigsqcap_{i=1}^{k}$\cone{(P_s, \mathcal{O}(P_{s+i}, \epsilon/2))} $\ne \{P_s\}$.
\end{prop}

\begin{proof}\
If $\bigsqcap_{i=s+1}^{e}{\mathcal{C}(P_s, P_{s+i}, \epsilon/2)} \ne \{P_s\}$, then by Proposition~\ref{prop-3d-ci}, there exists a point $Q$, $Q.t = P_{s+k}.t$, such that $sed(P_{s+i}, \vv{P_sQ})\le \epsilon/2$ for all $i \in [1,k]$. By the triangle inequality essentially, $sed(P_{s+i}, \vv{P_sP_{s+k}})\le  sed(P_{s+i}, \vv{P_sQ}) + |\vv{QP_{s+k}}| \le  \epsilon/2+\epsilon/2 = \epsilon$. \eop
\end{proof}

 We first present the  one-pass error bounded {\em strong trajectory simplification} algorithm using \sed, as shown in Figure~\ref{alg:CI3d}.



\stitle{Procedure \kw{getRegularPolygon}}.
We first present procedure \kw{getRegularPolygon} that, given a cone projection circle, generates its inscribed $m$-edge regular polygon,  following the definition in Section~\ref{subsec-RPI}.

The parameters $P_s$, $P_i$, $r$ and $t_c$ together form the projection circle \pcircle{(P^c_i, r^c_i)} of the spatio-temporal cone \cone{(P_s, \mathcal{O}(P_{i}, r))} of point $P_{i}$ \wrt point $P_s$ on the plane $P.t - t_c$ = $0$. Firstly, $P^c_i.x$ and $P^c_i.y$ are computed (lines 1--3), and $r^c_i = c\cdot r$.
Then it builds and returns an $m$-edge inscribed regular polygon $\mathcal{R}$ of \pcircle{(P^c_i, r^c_i)} (lines 4--8), by transforming a polar coordinate system
into a Cartesian one. Note here $\theta$, $r\cdot\sin\theta$ and $r\cdot\cos\theta$ only need to be computed once during the entire processing of a trajectory.

\stitle{Algorithm \cist}.
It takes as input a trajectory \trajec{T}${[P_0, \ldots, P_n]}$, an error bound $\epsilon$ and the number $m$ of edges for inscribed regular polygons, and returns a simplified  trajectory $\overline{\mathcal{T}}$ of $\dddot{\mathcal{T}}$.

\begin{figure*}[tb!]
	\centering
	\includegraphics[scale=0.79]{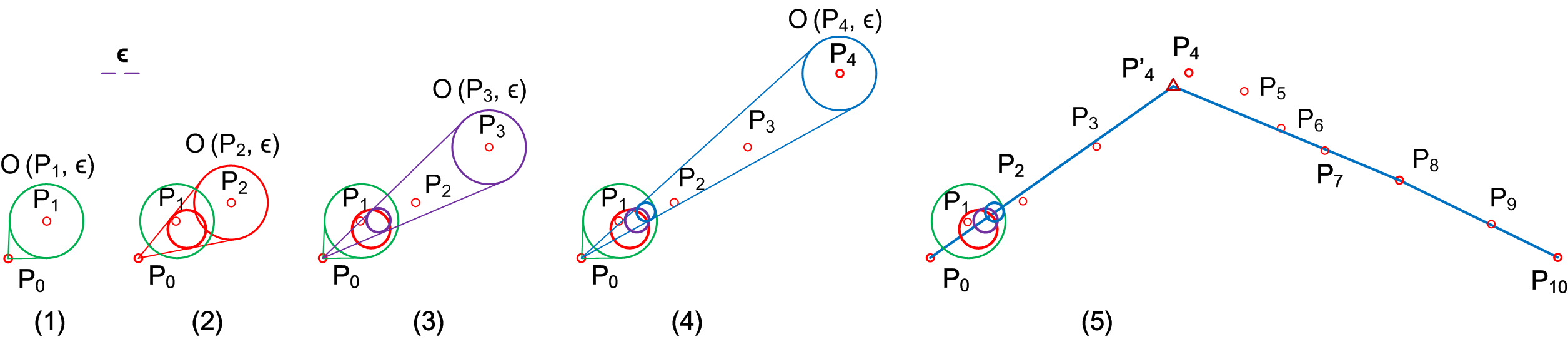}
	\caption{\small A running example of the \cista algorithm. The points and the oblique circular cones are projected on an x-y space. The trajectory $\dddot{\mathcal{T}}[P_0, \ldots, P_{10}]$ is compressed into three line segments.}
	\label{fig:exm-consta}
\end{figure*}

The algorithm first initializes the start point $P_s$ to $P_0$, the index $i$ of the current data point to $1$, the intersection polygon $\mathcal{R}^*$ to $\emptyset$, the output $\overline{\mathcal{T}}$ to $\emptyset$, and $t_c$ to $P_1.t$, respectively (line 1).
The algorithm sequentially processes the data points of the trajectory one by one  (lines 2--10). It gets the $m$-inscribed regular polygon \wrt the current point $P_i$ (line 3) by calling procedure $\kw{getRegularPolygon}$. When $\mathcal{R}^* = \emptyset$, the intersection polygon $\mathcal{R}^*$ is simply initialized as $\mathcal{R}$ (lines 4, 5). Otherwise,
$\mathcal{R}^*$ is  the intersection of the current regular polygon $\mathcal{R}$ with $\mathcal{R}^*$ by calling procedure $\rpia()$ introduced in Section~\ref{subsec-fastRPI} (line 7). If the resulting intersection $\mathcal{R}^*$ is empty, then a new line segment $\vv{P_sP_{i-1}}$ is generated (lines 8--10). The process repeats until all points have been processed (line 11).
After the  final new line segment $\vv{P_sP_{n}}$ is generated (line 12), it returns the simplified  piece-wise line representation $\overline{\mathcal{T}}$ (line 13).

\begin{figure}[tb!]
	\begin{center}
		{\small
			\begin{minipage}{3.3in}
				\myhrule
				\mat{0ex}{
					{\bf Algorithm}~\cist$(\dddot{\mathcal{T}}[P_0,\ldots,P_n],~\epsilon,~m)$\\
					\bcc \hspace{1ex}\= $P_s := P_0$; ~$i := 1$;  ~$\mathcal{R}^* := \emptyset$;  ~$\overline{\mathcal{T}} := \emptyset$; ~$t_c$ := $P_1.t$;\\
					\icc \hspace{1ex}\= \While $i \le n$ \Do \\
					\icc \>\hspace{2ex} $\mathcal{R}$ := \kw{getRegularPolygon}($P_s$, $P_i$, $\epsilon/2$, $m$, $t_c$); \\
					\icc \>\hspace{2ex} \If $\mathcal{R}^* = \emptyset$  \Then \hspace{3ex} /* $\mathcal{R}^*$ needs to be initialized */\\
					\icc \>\hspace{2ex} \ \ \ $\mathcal{R}^* :=\mathcal{R};$ \\
					\icc \>\hspace{2ex} \Else \\
					\icc \>\hspace{2ex}\ \ \ \ $\mathcal{R}^* := {\rpia}(\mathcal{R}^*, ~\mathcal{R})$; \\
					\icc \>\hspace{4ex} \If $\mathcal{R}^* = \emptyset$ \Then \hspace{1ex} /* generate a new line segment */\\
					\icc \> \hspace{4ex} \ \ \ $i$ := $i - 1$;\ \ $\overline{\mathcal{T}} := \overline{\mathcal{T}}\cup \{\vv{P_sP_{i}}\}$; \\
					\icc \> \hspace{4ex} \ \ \ $P_s := P_{i}$; \ \ $t_c$ := $P_{i+1}.t$;\\
					\icc \>\hspace{2ex} $i$ := $i +1$;\\
					\icc \> \hspace{0ex} $\overline{\mathcal{T}}$ := $\overline{\mathcal{T}}\cup \{\vv{P_sP_{n}}\}$; \\
					\icc \hspace{1ex}\Return $\overline{\mathcal{T}}$. \\
					\\
					{\bf Procedure} ~\kw{getRegularPolygon}$(P_s,~P_i,~r,~m,~t_c)$ \\
					\bcc \hspace{1ex} $c := (t_c-t_s)/(P_i.t - P_s.t)$; \\
					\icc \hspace{1ex} $x := P_s.x + c\cdot(P_i.x-P_s.x)$; \\
					\icc \hspace{1ex} $y := P_s.y + c\cdot(P_i.y-P_s.y)$; \\
					\icc \hspace{1ex} \For $(j := 1;j \le m;j++)$ \Do \\
					\icc \> \hspace{2ex} $\theta :=  (2j + 1)*\pi /m $; \\
					\icc \> \hspace{2ex} $\mathcal{R}.v_j.x := x + c\cdot r\cdot\cos\theta$;\\
					\icc \> \hspace{2ex} $\mathcal{R}.v_j.y := y + c\cdot r\cdot\sin\theta$;\\
					\icc \hspace{1ex} \Return $\mathcal{R}$.
				}
				\vspace{-2ex}
				\myhrule
			\end{minipage}
		}
	\end{center}
	\vspace{-1ex}
	\caption{\small One-pass strong trajectory  simplification algorithm.}
	\label{alg:CI3d}
	\vspace{-1ex}
\end{figure}

\begin{example}
\label{exm-alg-conest}
Figure~\ref{fig:exm-const} shows a running example of \cist for compressing the trajectory \trajec{T} in Figure~\ref{fig:notations}.

\sstab (1) After initialization, the \cist algorithm reads point $P_1$ and builds a narrow \emph{oblique circular cone} \cone{(P_0, \mathcal{O}(P_{1}, \epsilon/2))}, taking $P_0$ as its apex and \circle{(P_1, \epsilon/2)} as its base (green dash). The \emph{circular cone} is projected on the plane $P.t-P_1.t=0$, and the inscribe regular polygon $\mathcal{R}_1$ of the projection circle is returned. As $\mathcal{R}^*$ is empty, $\mathcal{R}^*$ is set to $\mathcal{R}_1$.

\sstab(2) The algorithm reads $P_2$ and builds \cone{(P_0, \mathcal{O}(P_{2}, \epsilon/2))} (red dash). The \emph{circular cone} is also projected on the plane $P.t-P_1.t=0$ and the inscribe regular polygon $\mathcal{R}_2$ of the projection circle is returned. As $\mathcal{R}^*=\mathcal{R}_1$ is not empty, $\mathcal{R}^*$ is set to the intersection of $\mathcal{R}_2$ and $\mathcal{R}^*$, which is $\mathcal{R}_1 \bigsqcap \mathcal{R}_2 \ne \emptyset$.

\sstab (3) For point $P_3$, the algorithm runs the same routine as $P_2$ until the intersection of $\mathcal{R}_3$ and $\mathcal{R}^*$ is $\emptyset$. Thus, a line segment $\vv{P_0P_2}$ is generated, and the process of a new line segment is started, taking $P_2$ as the new start point and $P.t-P_3.t=0$ as the new projection plane.

\sstab (4) At last, the algorithm outputs four continuous line segments, \ie $\{\vv{P_0P_2}$, $\vv{P_2P_4}$, $\vv{P_4P_{7}}$, $\vv{P_7P_{10}}\}$. \eop
\end{example}

\subsection{Weak Trajectory Simplification}

We then present the one-pass error bounded {\em weak simplification} algorithm using \sed.


\stitle{Algorithm \cista}.
Given a trajectory \trajec{T}${[P_0, \ldots, P_n]}$, an error bound $\epsilon$ and the number $m$ of edges for inscribed
regular polygons, it returns a simplified trajectory,
which may contain interpolated points.
By Proposition~\ref{prop-circle-intersection}, algorithm \cista generates spatio-temporal cones whose bases are circles with a radius of $\epsilon$,
and, hence, it replaces $\epsilon/2$ with $\epsilon$ (line 3 of \cist). It also generates new line segments with data points $Q$ (may be interpolated points), and,
hence, it replaces point $P_i$ and line segment $\vv{P_sP_i}$  (lines 9 and 10 of algorithm \cist) with $Q$ and $\vv{P_sQ}$, respectively,  such that $Q$ is generated as follows.

\begin{prop}
\label{prop-cist-Q}
Given a sub-trajectory \trajec{T}${[P_s, \ldots, P_{s+k}]}$ and an error bound $\epsilon$,  $t_c=P_{s+k}.t$ and $\mathcal{R}^*_k$ be the intersection of all polygons $\mathcal{R}_{s+i}$ ($i\in[1,k]$) on the plane $P.t - t_c = 0$. If $\mathcal{R}^*_k$ is not empty, then any point in the area of $\mathcal{R}^*_k$ is feasible for $Q$.
\end{prop}

\begin{proof}\
By Proposition~\ref{prop-circle-intersection} and the nature of inscribed regular polygon, it is easy to find that for any point $Q \in \mathcal{R}^*_k$  \wrt plane $t_c=P_{s+k}.t$, there is $sed(P_{s+i}, \vv{P_sQ})\le \epsilon$ for all points $P_{s+i}$ ($i \in [1,k]$).

From this, we have the conclusion. \eop
\end{proof}

The choice of a point $Q$ from $\mathcal{R}^*_k$ may slightly affect the effectiveness (\eg average errors and compression ratios). However, the choice of an optimal $Q$ is non-trivial. For the benefit of efficiency, we apply the following strategies.

\sstab (1) If $P_{s+k}$ is in the area of $\mathcal{R}^*_k$ \wrt $t_c=P_{s+k}.t$, then $Q$ is simply chosen as $P_{s+k}$.

\sstab (2) If $\mathcal{R}^*_k \ne \emptyset$ and $P_{s+k}$ is not in the area of $\mathcal{R}^*_k$ \wrt $t_c$=$P_{s+k}.t$, then the central point of $\mathcal{R}^*_k$ is chosen as $Q$.

\sstab (3)  If $t_c \ne P_{s+k}.t$, which is the general case, then we project the intersection polygon $\mathcal{R}^*_k$ \wrt $t_c \ne P_{s+k}.t$ on the plane $P.t -P_{s+k}.t = 0$, and apply strategies (1) and (2) above. That is, the projection has no affects on the choice of $Q$.

\begin{example}
\label{exm-alg-conesta}
Figure~\ref{fig:exm-consta}  shows a running example of algorithm \cista for compressing the trajectory \trajec{T} in Figure~\ref{fig:notations} again.

\sstab (1) After initialization, the \cista algorithm reads point $P_1$ and builds an \emph{oblique circular cone} \cone{(P_0, \mathcal{O}(P_{1}, \epsilon))}, and projects it on the plane $P.t-P_1.t=0$. The inscribed regular polygon $\mathcal{R}_1$ of the projection circle is returned and the intersection $\mathcal{R}^*$ is set to $\mathcal{R}_1$.

\sstab (2) $P_2$, $P_3$ and $P_4$ are processed in turn. The intersection polygons $\mathcal{R}^*$ are not empty.

\sstab (3) For point $P_5$, the intersection of polygons $\mathcal{R}_5$ and $\mathcal{R}^*$ is $\emptyset$. Thus, line segment $\vv{P_0Q} =\vv{P_0P'_4}$ is output, and a new line segment is started such that point $Q=P'_4$ is the new start point and plane $P.t-P_5.t=0$ is the new projection plane.

\sstab (4) At last, the algorithm outputs 3 continuous line segments, \ie $\vv{P_0P'_4}$, $\vv{P'_4P_8}$ and $\vv{P_8P_{10}}$, in which $P'_4$ is an interpolated data points not in \trajec{T}. \eop
\end{example}

\begin{table*}[bt!]
	\vspace{-1ex}
	\caption{\small Real-life trajectory datasets}
	\centering
	\small
	\begin{tabular}{|l|c|c|c|r|}
		\hline
		\bf{ Data Sets}& \bf{Number\ of Trajectories}     &\bf {Sampling Rates\ (s)}   &\bf{Points Per Trajectory\ (K)}    &\bf {Total points} \\
		\hline
		\sercar	&1,000	    &3-5	    &$\sim114.0$   &114M\\
		\hline
		\geolife &182	    &1-5	    &$\sim131.4$   &24.2M\\
		\hline
		\mopsi	&51	    	&2	    &$\sim153.9$     &7.9M\\
		\hline
		\pricar	& 10	    &1	        &$\sim11.8$      &112.8K \\
		\hline
	\end{tabular}
	\label{tab:datasets}
\end{table*}

\vspace{.5ex}
\stitle{Correctness and complexity analyses}.
The correctness of algorithms \cist and \cista follows from Propositions~\ref{prop-circle-intersection} and~\ref{prop-3d-ci-half}, and Propositions~\ref{prop-circle-intersection} and~\ref{prop-cist-Q}, respectively.
It is easy to verify that each data point in a trajectory is only processed once, and each can be done in $O(1)$ time,
as both procedures  $\kw{getRegularPolygon}$ and $\rpia$ can be done in $O(1)$ time.
Hence, these algorithms are both one-pass error bounded trajectory simplification algorithms.
It is also easy to see that these algorithms take $O(1)$ space.

\eat{
\begin{figure}[tb!]
\begin{center}
{\small
\begin{minipage}{3.36in}
\myhrule
\vspace{-1ex}
\mat{0ex}{
	{\bf Algorithm}~$\cista(\dddot{\mathcal{T}}[P_0,\ldots,P_n], ~\epsilon, ~M, ~t_c)$\\
	\bcc \hspace{1ex}\= $P_s := P_0$;  ~$P_e := P_0$;  ~$\overline{\mathcal{T}} := \phi$;  ~$\mathcal{R}^* := \phi$;  ~$i := 1$\\
	\icc \hspace{1ex}\= \While $i \le n$ \Do \\
	\icc \>\hspace{3ex} $\mathcal{R} := {getRegularPolygon}$($P_s$, $P_i$, $\epsilon$, $M$, $t_c$) \\
	\icc \>\hspace{3ex} \If $\mathcal{R}^* = \phi$ \Then \\
	\icc \>\hspace{6ex} $\mathcal{R}^* :=\mathcal{R}$ \\
	\icc \>\hspace{3ex} \Else \\
	\icc \>\hspace{6ex} $\mathcal{R}^* := {\rpia}(\mathcal{R}^*, ~\mathcal{R})$ \\
	\icc \>\hspace{3ex} \If $\mathcal{R}^* \ne \phi$ \Then \\
	\icc \> \hspace{6ex} $P_e := P_i$ \\
	\icc \> \hspace{6ex} $i := i+1$ \\
	\icc \>\hspace{3ex} \Else\\
	\icc \> \hspace{6ex} $\overline{\mathcal{T}} := \overline{\mathcal{T}}\cup \{\mathcal{L}(P_s,Q)\}$ \\
	\icc \> \hspace{6ex} $P_s := Q$;  ~~$\mathcal{R}^* := \phi$ \\
	\icc \hspace{1ex}\Return $\overline{\mathcal{T}}$
%
%
}
\vspace{-2ex}
\myhrule
\end{minipage}
}
\end{center}
\vspace{-2ex}
\caption{\small Aggressive spatio-temporal cone intersection algorithm (\cista).}
\label{alg:ciseda}
\vspace{-2ex}
\end{figure}
}

\section{Experimental Study} %
\label{sec-exp}

In this section, we present an extensive experimental study of our one-pass trajectory simplification algorithms (\cist and \cista) compared with the optimal algorithm using \sed and existing algorithms of \dps and \squishe on trajectory datasets.
Using four real-life trajectory datasets, we conducted three sets of experiments to evaluate:
(1) the compression ratios of algorithms \cist and \cista vs. \dps, \squishe and the optimal algorithm,
(2) the average errors of algorithms \cist and \cista vs. \dps, \squishe and the optimal algorithm,
(3) the execution time of algorithms \cist and \cista vs. \dps, \squishe and the optimal algorithm, and
(4) the impacts of polygon intersection algorithms \rpia and \cpia and the edge number $m$ of inscribed regular polygons to the effectiveness and efficiency of algorithms \cist and \cista.

\subsection{Experimental Setting}

\stitle{Real-life Trajectory Datasets}.
We use four real-life datasets \sercar, \geolife, \mopsi and \pricar shown in Table~\ref{tab:datasets} to test our solutions.


\vspace{0.5ex}
\ni \emph{(1) Service car trajectory data} (\sercar) is the GPS trajectories collected by a Chinese car rental company during Apr. 2015 to Nov. 2015. The sampling rate was one point per $3$--$5$ seconds, and
each trajectory has around $114.1K$ points.

\vspace{0.5ex}
\ni \emph{(2) GeoLife trajectory data} (\geolife) is the GPS trajectories collected in GeoLife project~\cite{Zheng:GeoLife} by 182 users in a period from Apr. 2007 to Oct. 2011. These trajectories have a variety of sampling rates, among which 91\% are logged in each 1-5 seconds per point. 

\vspace{0.5ex}
\ni \emph{(3) Mopsi trajectory data} (\mopsi) is the GPS trajectories collected in Mopsi project~\cite{Mopsi} by 51 users in a period from 2008 to 2014. Most routes are in Joensuu region, Finland.
The sampling rate was one point per $2$ seconds, and each trajectory has around $153.9K$ points.

\vspace{0.5ex}
\ni \emph{(4) Private car trajectory data} (\pricar) is a small set GPS trajectories collected with a high sampling rate of one point per second by our team members in 2017. There are 10 trajectories and each trajectory has around 11.8K points.


\vspace{0.5ex}
\ni \emph{(5) Small trajectory data}.
As the optimal \lsa algorithm~\cite{Imai:Optimal} it has both high time and space complexities, \ie $O(n^3)$ time and $O(n^2)$ space, it is impossible to compress the entire datasets (too slow and out of memory). Hence, we further build four \textit{small datasets}, each dataset includes 10 middle-size ($10K$ points per trajectory) trajectories selected from \sercar, \geolife, \mopsi and \pricar, respectively.


\stitle{Algorithms and implementation}.
We implement five \lsa algorithms, \ie our \cist and \cista, \dps~\cite{Meratnia:Spatiotemporal} (the most effective existing \lsa algorithm using \sed), \squishe~\cite{Muckell:Compression} (the most efficient existing \lsa algorithm using \sed) and the optimal \lsa algorithm using \sed (see Section~\ref{subsec-optimal}).
We also implement the polygon intersection algorithms, \cpia and our \rpia.
All algorithms were implemented with Java.
All tests were run on an {x64-based  PC with 8 Intel(R) Core(TM) i7-6700 CPU @ 3.40GHz and 8GB of memory, and each test was repeated
over 3 times and the average is reported here.

%

\subsection{Experimental Results}

We next present our findings.

\subsubsection{Evaluation of Compression Ratios}

In the first set of tests, we evaluate the impacts of parameter $m$ on the
compression ratios of our algorithms \cist and \cista, and compare the compression ratios of \cist and \cista with \dps, \squishe and the optimal algorithm.
The compression ratio is defined as follows: Given a set of trajectories $\{\dddot{\mathcal{T}_1}, \ldots, \dddot{\mathcal{T}_M}\}$ and their piecewise line representations $\{\overline{\mathcal{T}_1}, \ldots, \overline{\mathcal{T}_M}\}$, the compression ratio of an algorithm is $(\sum_{j=1}^{M} |\overline{\mathcal{T}}_j |)/(\sum_{j=1}^{M} |\dddot{\mathcal{T}}_j |)$.
By the definition, \emph{algorithms with lower compression ratios are better}.


\begin{figure*}[tb!]
\centering
\includegraphics[scale = 0.290]{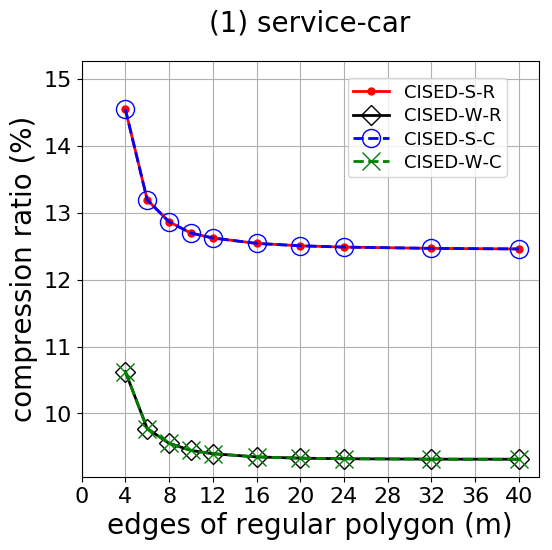}\hspace{1ex}
\includegraphics[scale = 0.290]{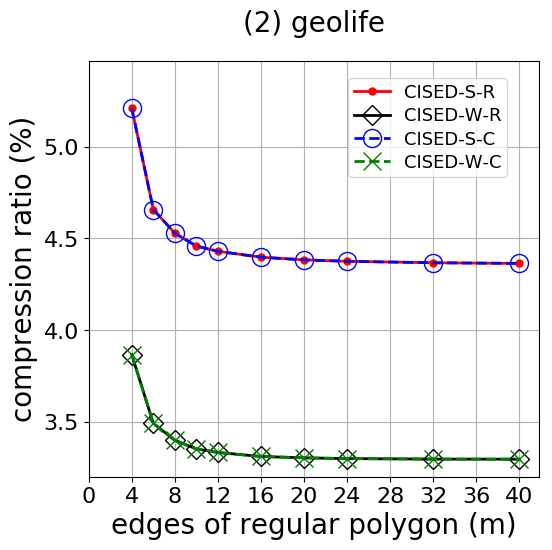}\hspace{1ex}
\includegraphics[scale = 0.290]{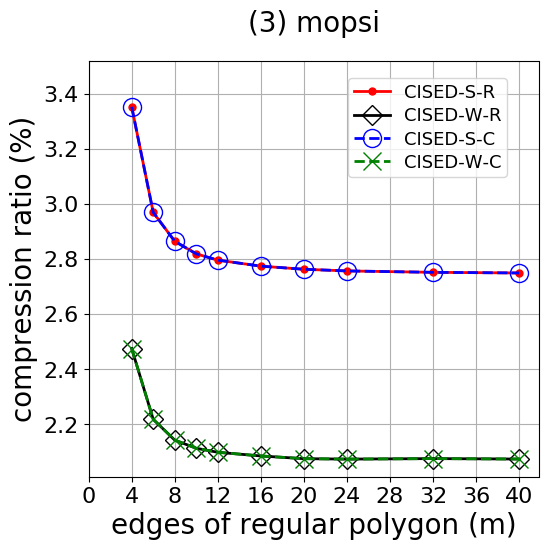}\hspace{1ex}
\includegraphics[scale = 0.290]{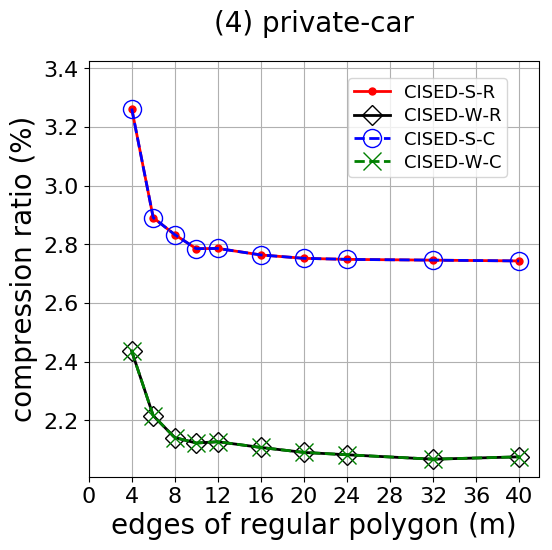}
\caption{\small Evaluation of compression ratios: fixed error bound with $\epsilon=60$ meters and varying $m$.}
\label{fig:m-cr-e60}
\end{figure*}

\begin{figure*}[tb!]
\centering
\includegraphics[scale = 0.290]{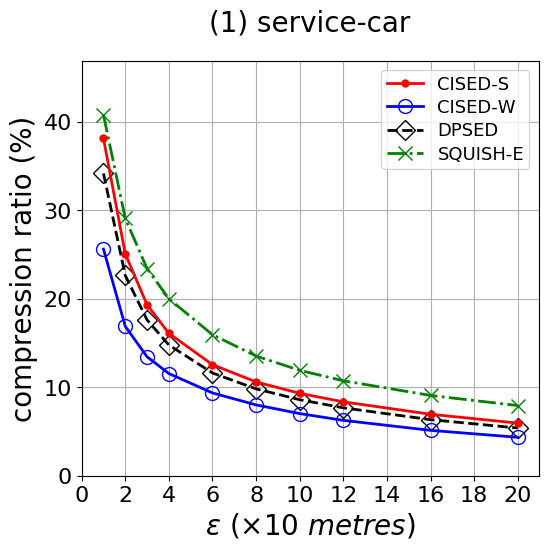}\hspace{1ex}
\includegraphics[scale = 0.290]{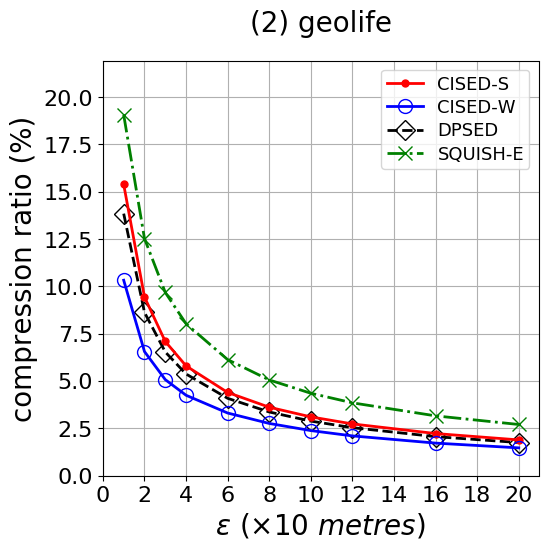}\hspace{1ex}
\includegraphics[scale = 0.290]{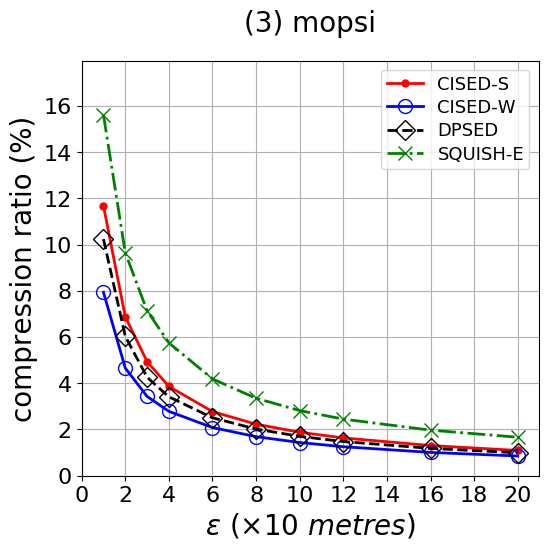}\hspace{1ex}
\includegraphics[scale = 0.290]{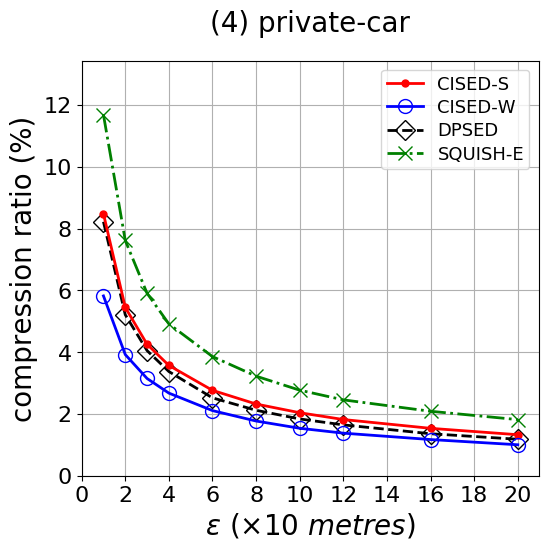}
\caption{\small Evaluation of compression ratios: fixed with $m=16$ and varying error bound $\epsilon$.}
\label{fig:cr-m16}
\end{figure*}

\begin{figure*}[tb!]
\centering
\includegraphics[scale = 0.290]{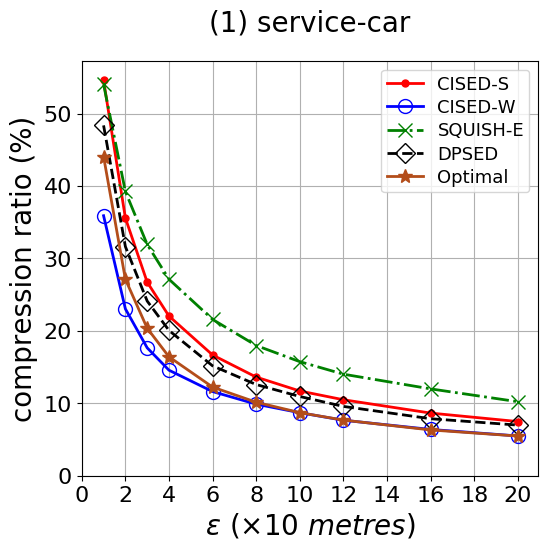}\hspace{1ex}
\includegraphics[scale = 0.290]{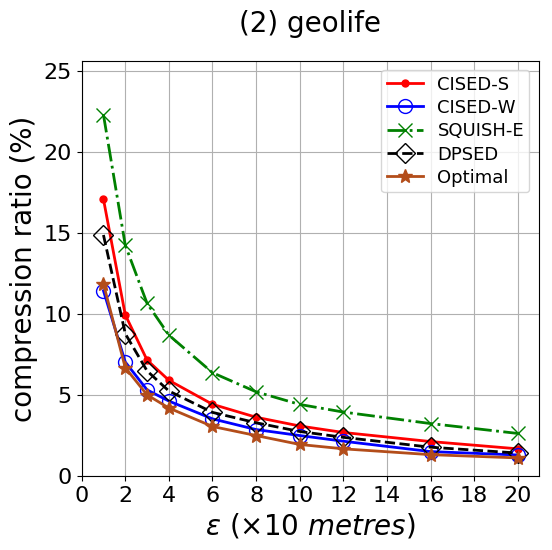}\hspace{1ex}
\includegraphics[scale = 0.290]{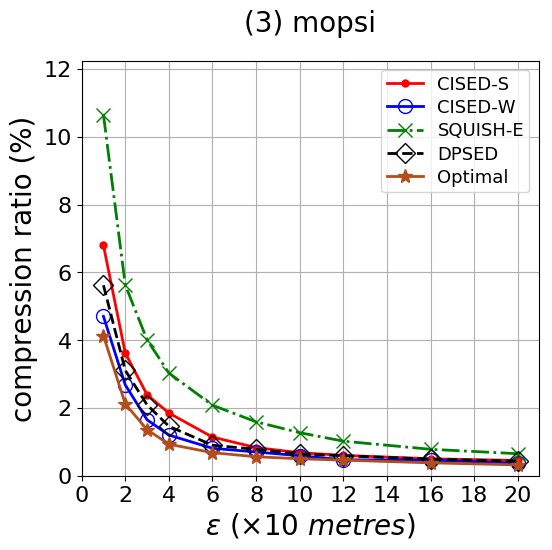}\hspace{1ex}
\includegraphics[scale = 0.290]{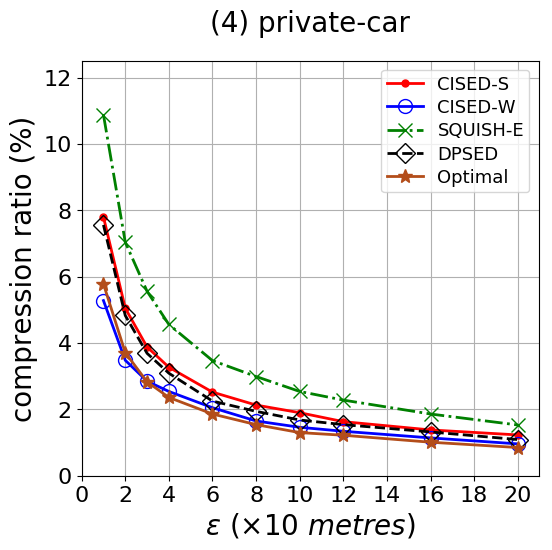}
\caption{\small Evaluation of compression ratios: fixed with $m=16$ and varying error bound $\epsilon$ (on small datasets).}
\label{fig:cr-optimal-m16}
\end{figure*}

\begin{figure*}[tb!]
\centering
\includegraphics[scale = 0.2900]{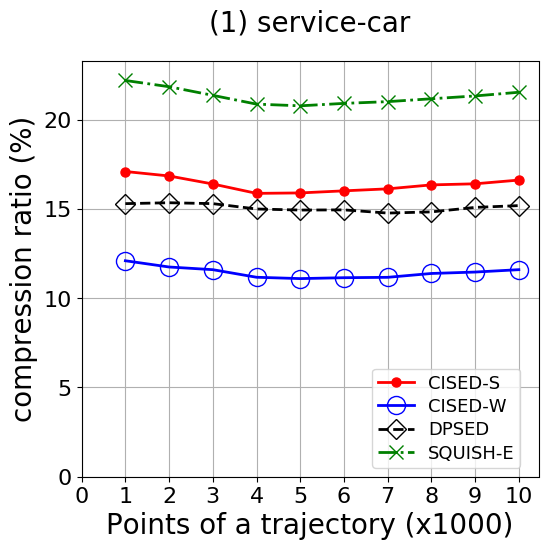}\hspace{1ex}
\includegraphics[scale = 0.2900]{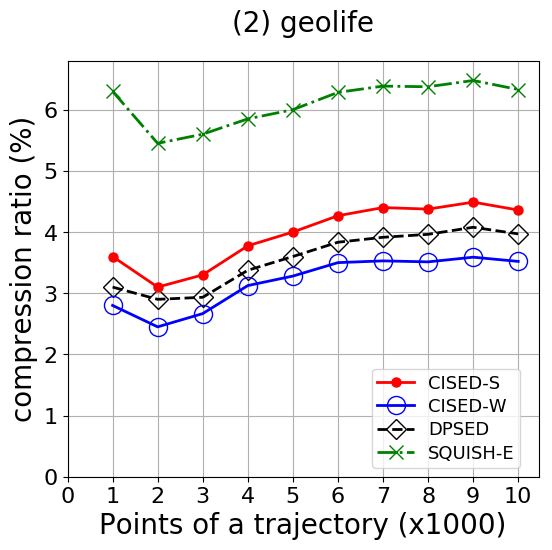}\hspace{1ex}
\includegraphics[scale = 0.2900]{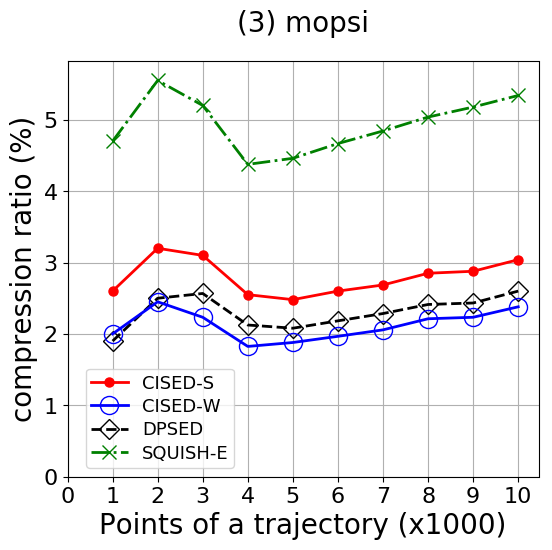}\hspace{1ex}
\includegraphics[scale = 0.2900]{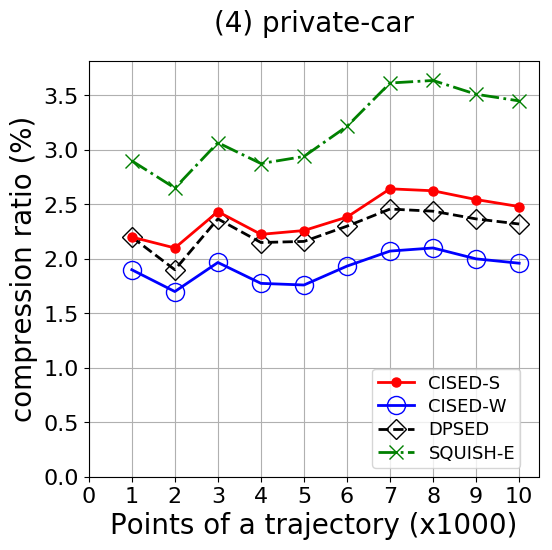}
\caption{\small Evaluation of compression ratios: fixed with $m=16$ and $\epsilon=60$ meters, and varying the size of trajectories.}
\label{fig:cr-size}
\end{figure*}

\stitle{Exp-1.1: Impacts of parameter $m$ on compression ratios}.
To evaluate the impacts of the number $m$ of edges of polygons on the
compression ratios of algorithms \cist and \cista, we fixed the error bounds {$\epsilon =60$ meters}, and varied $m$ from $4$ to $40$. The results are reported in Figure~\ref{fig:m-cr-e60}.

\ni(1) Algorithms \cist and \cista using \rpia have the same compression ratios as their counterparts using \cpia for all cases.

\ni(2) When varying $m$, the compression ratios of algorithms
{\cist and \cista} decrease with the increase of $m$ on all datasets.

\ni(3) When varying $m$, the compression ratios of algorithms {\cist and \cista} decrease (a) fast when $m < 12$, (b) slowly when $m \in [12, 20]$, and (c) very slowly when $m > 20$. Hence, \emph{the region of $[12, 20]$ is the good candidate region for $m$ in terms of compression ratios.}
Here the compression ratio of $m$=$12$ is only on average {$100.95\%$} of $m$=$20$.


\stitle{Exp-1.2: Impacts of the error bound $\epsilon$ on compression ratios (VS. algorithms \dps and \squishe)}.
To evaluate the impacts of error bound $\epsilon$ on compression ratios, we fixed {$m$=$16$}, the middle of $[12, 20]$, and varied $\epsilon$ from $10$ meters to $200$ meters on the entire four datasets, respectively.
The results are reported in Figure~\ref{fig:cr-m16} .

\ni (1) When increasing $\epsilon$, the compression ratios of all these algorithms decrease on all datasets.

\ni (2) Dataset \pricar has the lowest compression ratios, compared with datasets \mopsi, \sercar and \geolife, due to its highest sampling rate,
\sercar has the highest compression ratios due to its lowest sampling rate, and \geolife and \mopsi have the compression ratios in the middle accordingly.

\ni {(3)} Algorithm \cist is better than \squishe {and comparable} with \dps on all datasets and for all $\epsilon$.
The compression ratios of \cist are on average {($79.3\%$, $71.9\%$, $67.3\%$, $72.7\%$) and ($109.2\%$, $108.0\%$, $111.7\%$, $109.1\%$)} of \squishe and
\dps on {datasets (\sercar, \geolife, \mopsi, \pricar)}, respectively.
For example, when $\epsilon$ = $40$ meters, the compression ratios of algorithms
\squishe, \cist and \dps are
{($20.0\%$, $8.0\%$, $5.7\%$, $4.9\%$), ($16.1\%$, $5.8\%$, $3.9\%$, $3.6\%$) and ($14.8\%$, $5.4\%$, $3.4\%$, $3.4\%$)} on  {datasets (\sercar, \geolife, \mopsi, \pricar)}, respectively.

\ni {(4)} Algorithm \cista has better compression ratios than \dpa, \squishe and \cist on all datasets and for all $\epsilon$.
The compression ratios of \cista are on average ($57.7\%$, $53.8\%$, $50.0\%$, $54.6\%$), ($79.5\%$, $81.0\%$, $83.0\%$, $82.0\%$) and {($72.9\%$, $75.0\%$, $74.3\%$, $75.1\%$) of algorithms
\squishe, \dps and \cist on {datasets (\sercar, \geolife, \mopsi, \pricar)}, respectively.
For example, when $\epsilon$ = $40$ meters, the compression ratios of algorithm
\cista are ($11.5\%$, $4.3\%$, $2.8\%$, $2.7\%$) on datasets (\sercar, \geolife, \mopsi, \pricar), respectively.

\stitle{Exp-1.3: Impacts of the error bound $\epsilon$ on compression ratios (VS. the optimal algorithm).}
To evaluate the impacts of error bound $\epsilon$ on compression ratios, we once again fixed {$m$=$16$}, the middle of $[12, 20]$, and varied $\epsilon$ from $10$ to $200$ meters on the first $1K$ points of each trajectory of the selected \textit{small datasets}, respectively.
The results are reported in Figure~\ref{fig:cr-optimal-m16} .

\ni {(1)} Algorithm \cist is poorer than the optimal algorithm on all datasets and for all $\epsilon$.
More specifically, the compression ratios of \cist are on average ($134.6\%$, $150.7\%$, $155.5\%$, $138.5\%$) of the optimal algorithm on {datasets (\sercar, \geolife, \mopsi, \pricar)}, respectively.
For example, when $\epsilon$ = $40$ meters, the compression ratios of \cist and the optimal algorithm are
($22.0\%$, $5.9\%$, $1.9\%$, $3.3\%$) and {($16.4\%$, $4.2\%$, $0.9\%$, $2.4\%$)}
on  {datasets (\sercar, \geolife, \mopsi, \pricar)}, respectively.

\ni {(2)} Algorithm \cista is comparable with the optimal algorithm on all datasets and for all $\epsilon$.
The compression ratios of \cista are on average  ($94.8\%$, $115.5\%$, $119.7\%$, $107.5\%$)} of the optimal algorithm
 on {datasets (\sercar, \geolife, \mopsi, \pricar)}, respectively.
For example, when $\epsilon$ = $40$ meters, the compression ratios of algorithm
\cista are ($14.6\%$, $4.6\%$, $1.2\%$, $2.5\%$) on datasets (\sercar, \geolife, \mopsi, \pricar), respectively.

\stitle{Exp-1.4: Impacts of trajectory sizes on compression ratios}.
To evaluate the impacts of trajectory size, \ie the number of data points in a trajectory, on compression ratios,
we chose the same {$10$} trajectories from datasets \sercar, \geolife, \mopsi and \pricar, respectively,
fixed {$m$=$16$} and $\epsilon$=$60$ meters, and varied the size \trajec{|T|} of trajectories from $1K$ points to $10K$ points.
The results are reported in Figure~\ref{fig:cr-size}.

\ni(1) The compression ratios of these algorithms from the best to the worst are \cista, \dps, \cist and \squishe, on all datasets and for all sizes of trajectories.

\ni(2) The size of input trajectories has few impacts on the compression ratios of \lsa algorithms on all datasets.

\subsubsection{Evaluation of Average Errors}


\begin{figure*}[tb!]
	\centering
	\includegraphics[scale = 0.2900]{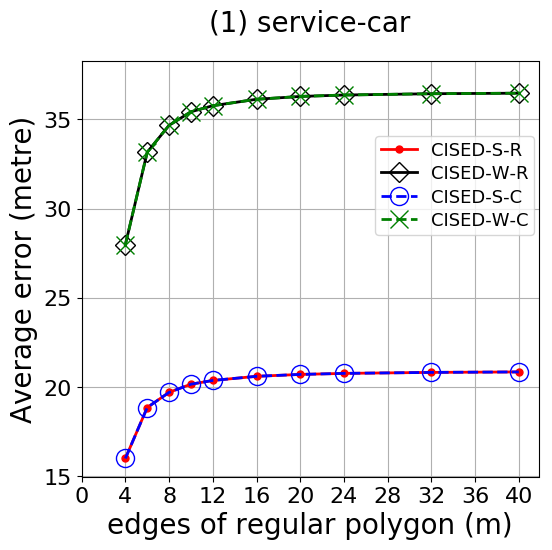}\hspace{1ex}
	\includegraphics[scale = 0.2900]{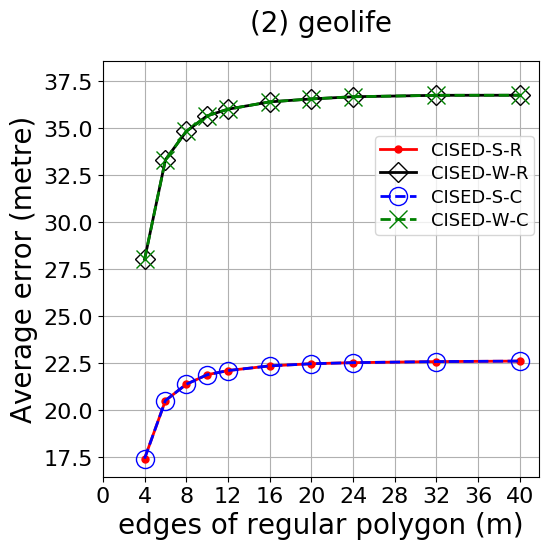}\hspace{1ex}
	\includegraphics[scale = 0.2900]{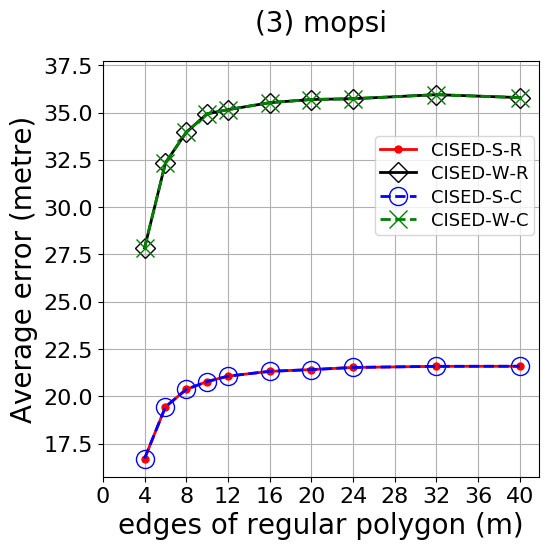}\hspace{1ex}
	\includegraphics[scale = 0.2900]{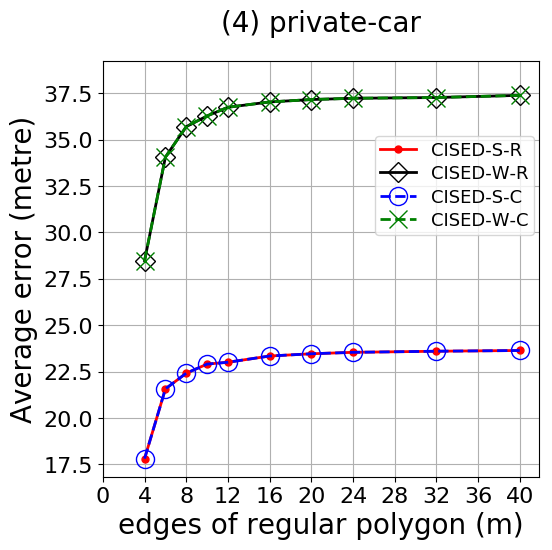}
	\caption{\small Evaluation of average errors: fixed error bound with $\epsilon = 60$ meters and varying $m$.}
	\label{fig:m-error-e60}
\end{figure*}

\begin{figure*}[tb]
	\centering
	\includegraphics[scale = 0.2900]{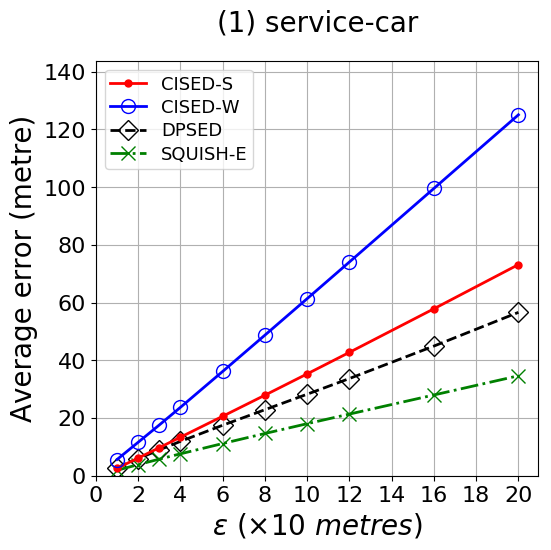}\hspace{1ex}
	\includegraphics[scale = 0.2900]{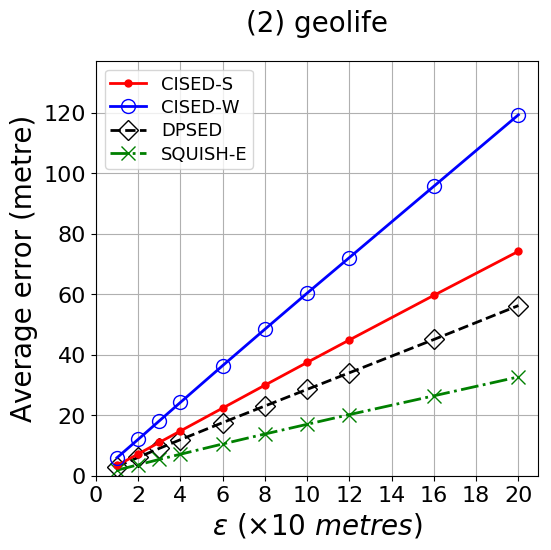}\hspace{1ex}
	\includegraphics[scale = 0.2900]{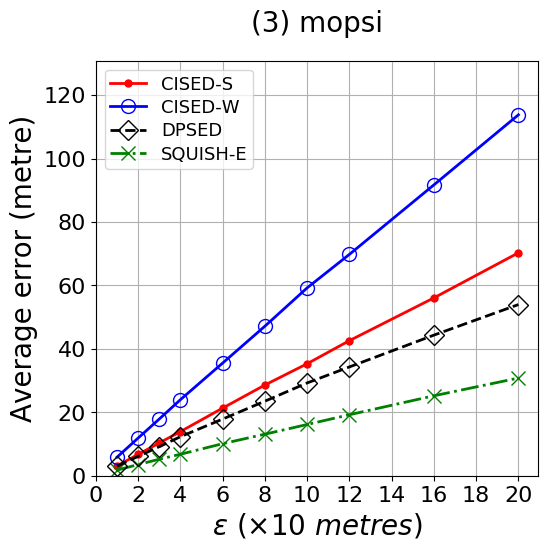}\hspace{1ex}
	\includegraphics[scale = 0.2900]{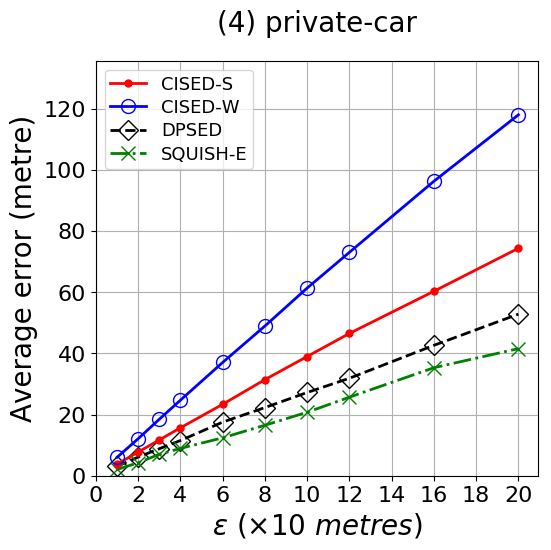}
	\caption{\small Evaluation of average errors: fixed with $m=16$ and varying error bound $\epsilon$.}
	\label{fig:ae-m16}
\end{figure*}

\begin{figure*}[tb]
	\centering
	\includegraphics[scale = 0.2900]{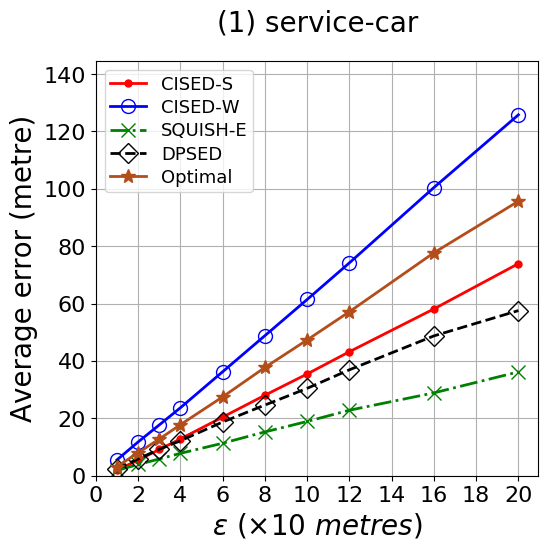}\hspace{1ex}
	\includegraphics[scale = 0.2900]{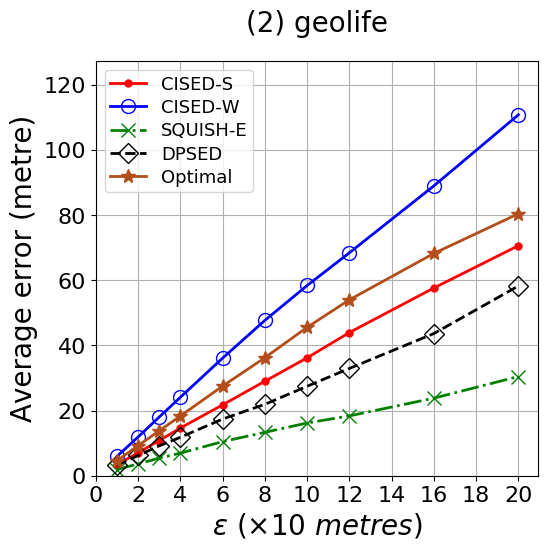}\hspace{1ex}
	\includegraphics[scale = 0.2900]{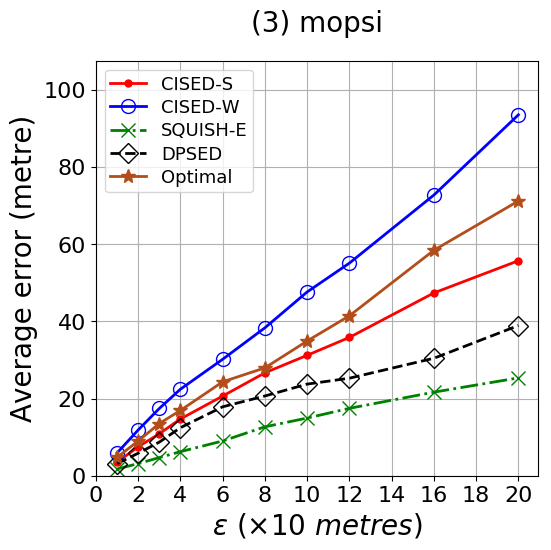}\hspace{1ex}
	\includegraphics[scale = 0.2900]{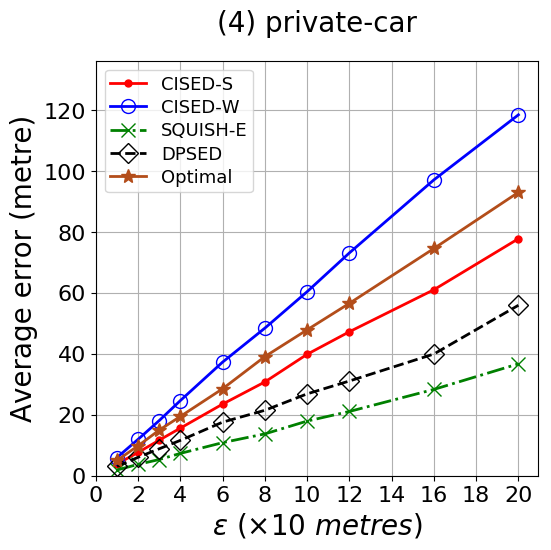}
	\caption{\small Evaluation of average errors: fixed with $m=16$ and varying error bound $\epsilon$ (on small datasets).}
	\label{fig:ae-optimal-m16}
\end{figure*}

\begin{figure*}[tb!]
	\centering
	\includegraphics[scale = 0.2900]{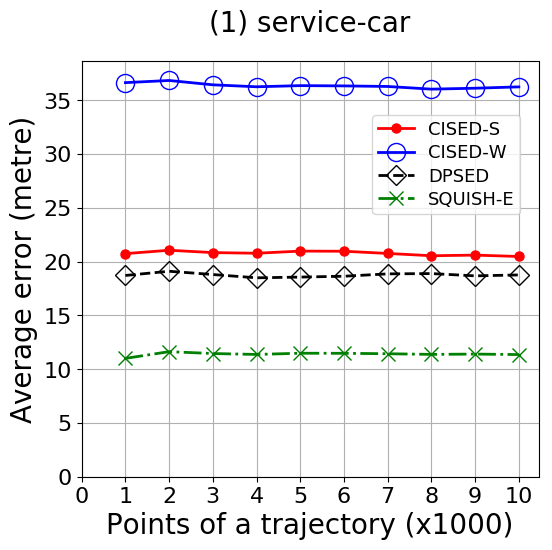}\hspace{1ex}
	\includegraphics[scale = 0.2900]{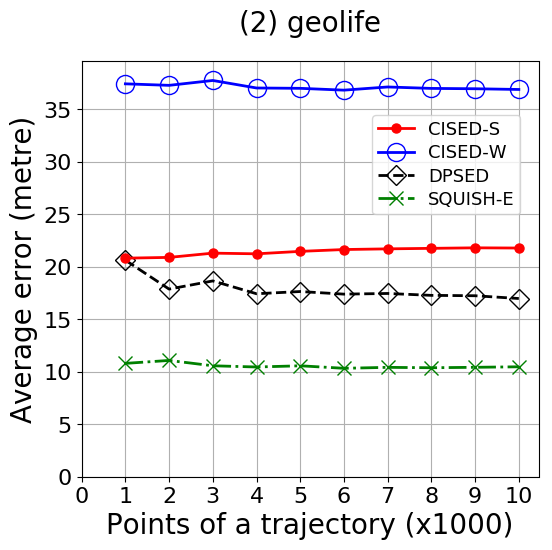}\hspace{1ex}
	\includegraphics[scale = 0.2900]{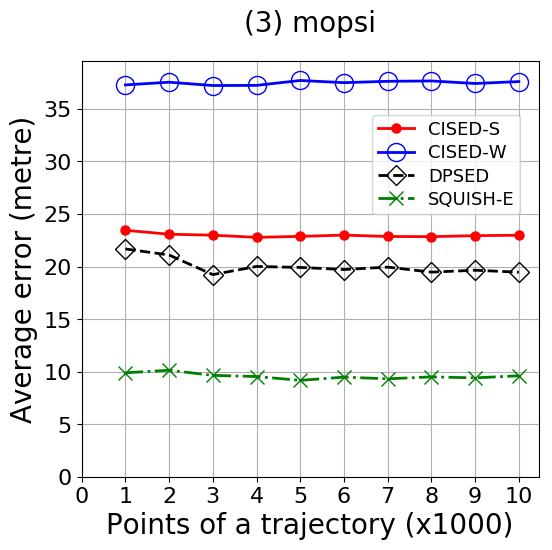}\hspace{1ex}
	\includegraphics[scale = 0.2900]{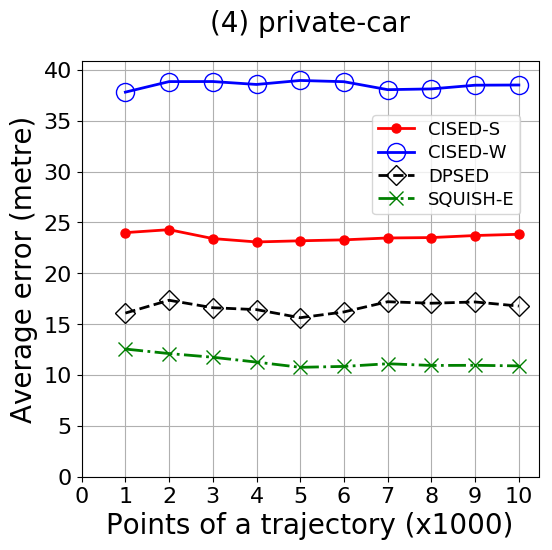}
	\caption{\small Evaluation of average errors: fixed with $m=16$ and $\epsilon=60$ meters, and varying the size of trajectories.}
	\label{fig:ae-size}
\end{figure*}

In the second set of tests, we first evaluate the impacts of parameter $m$ on the average errors of algorithms \cist and \cista, then compare the average errors of our algorithms \cist and \cista with \dps, \squishe and the optimal algorithm.

Given a set of trajectories $\{\dddot{\mathcal{T}_1}, \ldots, \dddot{\mathcal{T}}_M\}$ and their piecewise line representations $\{\overline{\mathcal{T}_1}, \ldots, \overline{\mathcal{T}}_M\}$, and point $P_{j,i}$ denoting
a point in trajectory $\dddot{\mathcal{T}}_j$ contained in a line segment $\mathcal{L}_{l,i}\in\overline{\mathcal{T}_l}$ ($l\in[1,M]$),
then the average error is $\sum_{j=1}^{M}\sum_{i=0}^{M} d(P_{j,i},
\mathcal{L}_{l,i})/\sum_{j=1}^{M}{|\dddot{\mathcal{T}}_j |}$.

\stitle{Exp-2.1: Impacts of parameter $m$ on average errors}.
To evaluate the impacts of parameter $m$ on average errors of algorithms \cist and \cista, we fixed the error bounds {$\epsilon =60$ meters}, and varied $m$ from $4$ to $40$. The results are reported in Figure~\ref{fig:m-error-e60}.


\ni(1) Algorithms \cist and \cista using \rpia have the same average errors as their counterparts using \cpia, respectively, on all datasets and for all $m$.

\ni(2) When varying $m$, the average errors of algorithms \cist and \cista increase with the increase of $m$ on all datasets.

\ni(3) When varying $m$, similar to compression ratios, the average errors of
algorithms \cist and \cista increase (a) fast when $m < 12$, (b) slowly when $m
\in [12, 20]$, and (c) very slowly when $m > 20$.
\emph{The range of $[12, 20]$ is also the good candidate region for $m$ in terms of errors.}
Here the average error of $m=12$ is only on average {$98.49\%$} of $m=20$.

\stitle{Exp-2.2: Impacts of the error bound $\epsilon$ on average errors (VS. algorithms \dps and \squishe)}.
To evaluate the average errors of these algorithms, we fixed {$m$=$16$}, and
varied $\epsilon$ from $10$ meters to $200$ meters on the entire
{datasets} \sercar, \geolife, \mopsi and \pricar, respectively.
The results are reported in Figure~\ref{fig:ae-m16}.

\ni(1) Average errors increase with the increase of $\epsilon$.

\ni(2) The average errors of these algorithms from the largest to the smallest are \cista, \cist, \dps and \squishe, on all datasets and for all $\epsilon$.
The average errors of algorithms \cist and \cista are on average
($119.3\%$, $127.7\%$, $119.9\%$, $138.0\%$)
and ($210.1\%$, $207.5\%$, $200.9\%$, $217.5\%$)
of \dps and ($188.2\%$, $215.2\%$, $212.8\%$, $180.3\%$) and
($331.1\%$, $349.7\%$, $356.7\%$, $284.2\%$)
 of \squishe on datasets (\sercar, \geolife, \mopsi, \pricar), respectively.

\ni(3) When the error bound of algorithm \cista is set as the half of \cist, the
average errors of \cista are on average ($93.8\%$, $86.0\%$, $81.4\%$, {$79.4\%$}) of \cist on {datasets} (\sercar, \geolife,\mopsi, \pricar), respectively, meaning that the large average errors of algorithm \cista are caused by its cone \wrt $\epsilon$ compared with the narrow cone \wrt $\epsilon/2$ of \cist.

\stitle{Exp-2.3: Impacts of the error bound $\epsilon$ on average errors (VS. the optimal algorithm).}
To evaluate the average errors of these algorithms, we once again fixed {$m$=$16$}, and
varied $\epsilon$ from $10$ meters to $200$ meters on the first $1K$ points of each trajectory of the selected \textit{small datasets}, respectively.
The results are reported in Figure~\ref{fig:ae-optimal-m16}.

The average errors of these algorithms from the largest to the smallest are \cista, the optimal algorithm and \cist, on all datasets and for all $\epsilon$.
The average errors of \cist and \cista are on average
($73.6\%$, $80.7\%$, $85.1\%$, $81.0\%$)
and ($133.3\%$, $130.7\%$, $131.0\%$, $126.3\%$)
of the optimal algorithm on datasets (\sercar, \geolife, \mopsi, \pricar), respectively.

\stitle{Exp-2.4: Impacts of trajectory sizes on average errors}.
To evaluate the impacts of trajectory sizes on average errors, we chose the same
{$10$} trajectories from  datasets \sercar, \geolife, \mopsi and \pricar, respectively.
We fixed {$m$=$16$} and $\epsilon = 60$ meters, and varied the size \trajec{|T|} of trajectories from $1K$ points to $10K$ points.
The results are reported in Figure~\ref{fig:ae-size}.

\ni(1) The average errors of these algorithms from the smallest to the largest are \squishe, \dps, \cist and \cista, on all datasets and for all trajectory sizes. 

\ni(2) The size of input trajectories has few impacts on the average errors of \lsa algorithms on all datasets.

\subsubsection{Evaluation of Efficiency}

In the last set of tests, we evaluate the impacts of parameter $m$ on the efficiency of algorithms \cist and \cista, and compare the efficiency of our approaches \cist and \cista with the optimal algorithm and algorithms \dps and \squishe.

%

\stitle{Exp-3.1: Impacts of algorithm \rpia and parameter $m$ on efficiency.}
To evaluate the impacts of \rpia and parameter $m$ on algorithm \cist and \cista, we
equipped \cist and \cista with \rpia and \cpia, respectively, fixed $\epsilon =60$ meters, and varied $m$ from $4$ to $40$.
The results are reported in Figures~\ref{fig:m-poly-time} and~\ref{fig:m-time-e60}.

\ni(1) The algorithms \cist and \cista spend the most time in the executing of
polygon intersections. For all $m$, the execution time of algorithms \cpia and
\rpia is on average {($93.5\%$, $96.0\%$, $94.5\%$, $92.0\%$)
	and ($90.5\%$, $92.5\%$, $91.0\%$, $90.5\%$)} of the entire compression  time on {datasets}
(\sercar, \geolife, \mopsi, \pricar), respectively.

\ni(2) \rpia runs faster than \cpia on all datasets and for all $m$. The execution time of algorithms \cist-\rpia and \cista-\rpia is one average $83.74\%$ their counterparts with \cpia.

\ni(3) When varying $m$, the execution time of algorithms \cist-\rpia, \cist-\cpia, \cista-\rpia and \cista-\cpia increases approximately linearly with the increase of $m$ on all the datasets.

\ni(4) The running time of $m=12$ is on average $69.92\%$ of $m=20$ for \cist and \cista on all datasets.



\begin{figure*}[tb!]
	\centering
	\includegraphics[scale = 0.290]{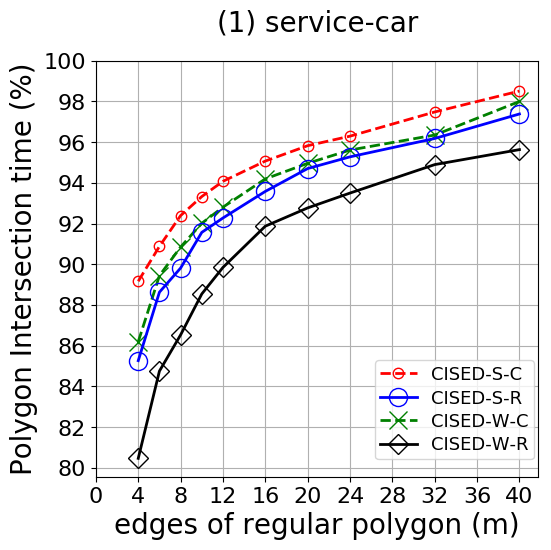}\hspace{1ex}
	\includegraphics[scale = 0.290]{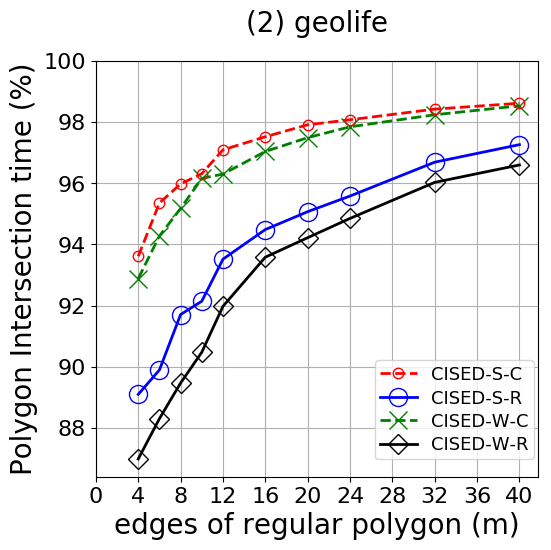}\hspace{1ex}
	\includegraphics[scale = 0.290]{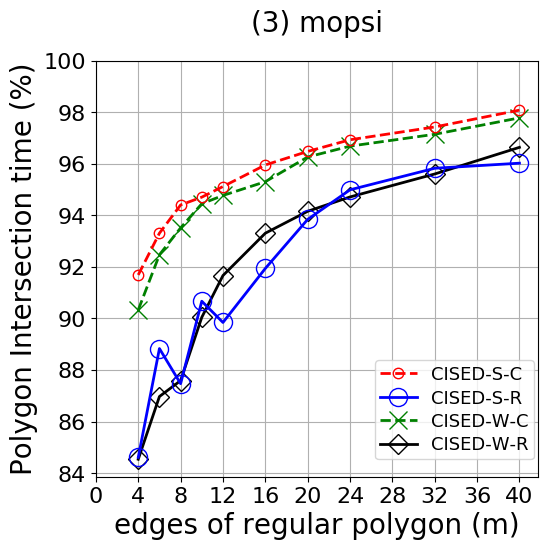}\hspace{1ex}
	\includegraphics[scale = 0.290]{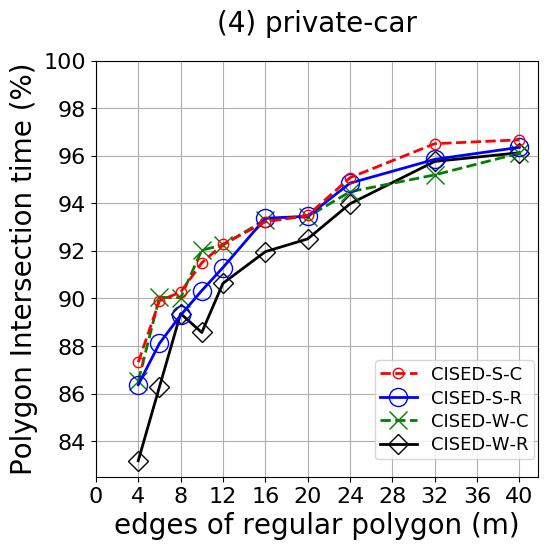}
	\caption{\small Evaluation of running time of polygon intersection algorithms: fixed error bound with $\epsilon=60$ meters, and varying $m$. Here ``R'' denotes our fast regular polygon intersection algorithm \rpia, and ``C'' denotes \cpia, respectively.}
	\label{fig:m-poly-time}
\end{figure*}

\begin{figure*}[tb!]
	\centering
	\includegraphics[scale = 0.290]{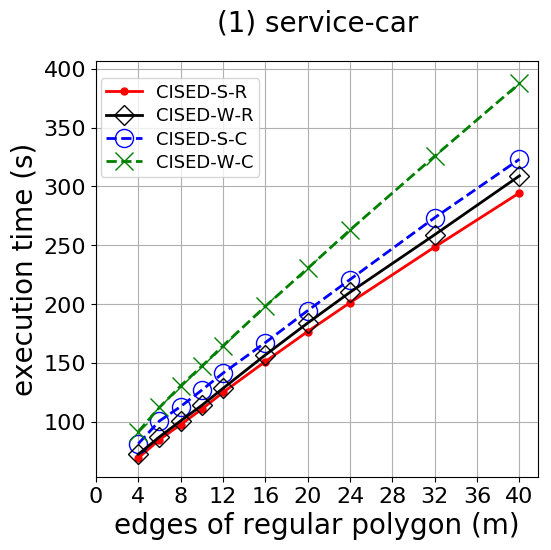}\hspace{1ex}
	\includegraphics[scale = 0.290]{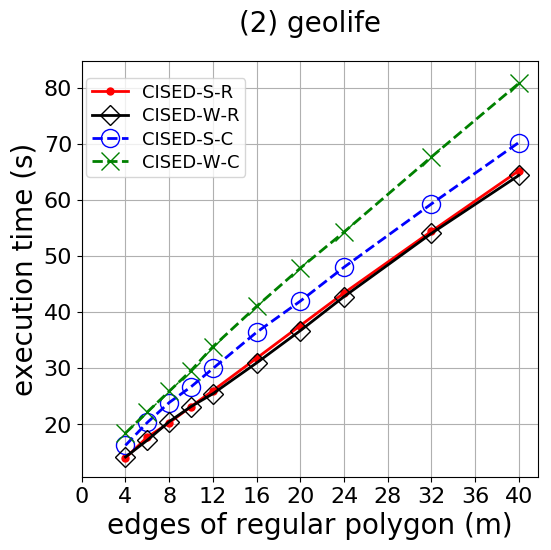}\hspace{1ex}
	\includegraphics[scale = 0.290]{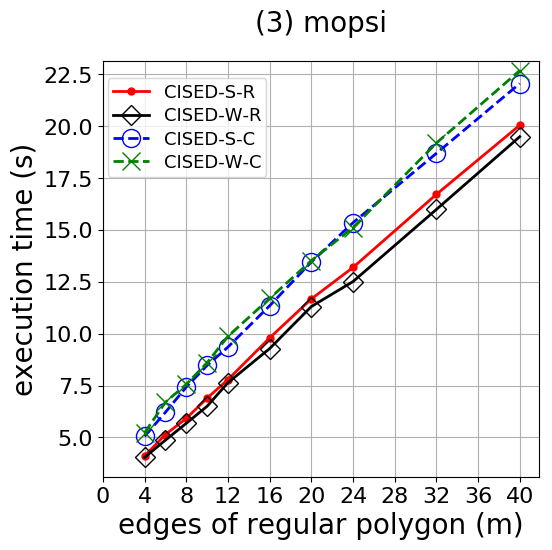}\hspace{1ex}
	\includegraphics[scale = 0.290]{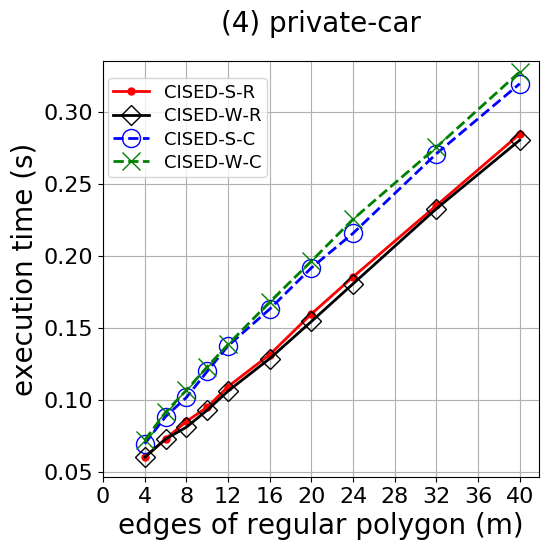}
	\caption{\small Evaluation of running time: fixed error bound with $\epsilon=60$ meters, and varying $m$. }
	\label{fig:m-time-e60}
\end{figure*}

\begin{figure*}[tb!]
\centering
\includegraphics[scale = 0.290]{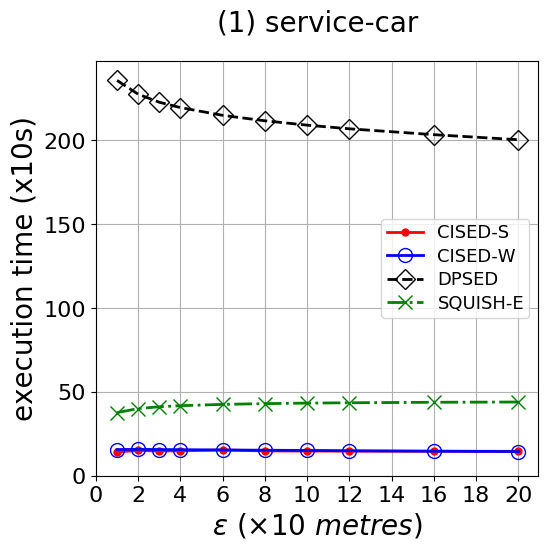}\hspace{1ex}
\includegraphics[scale = 0.290]{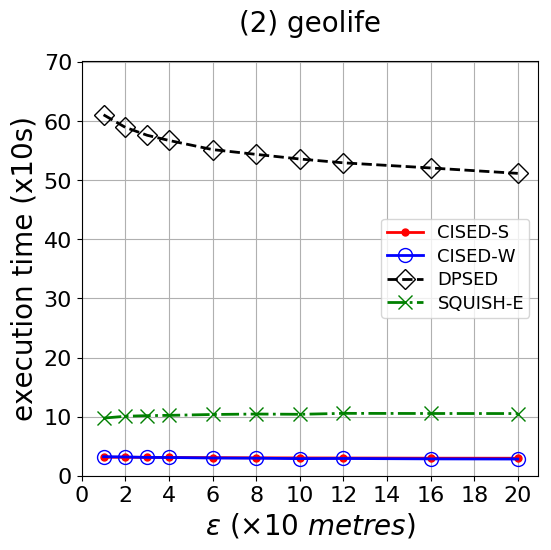}\hspace{1ex}
\includegraphics[scale = 0.290]{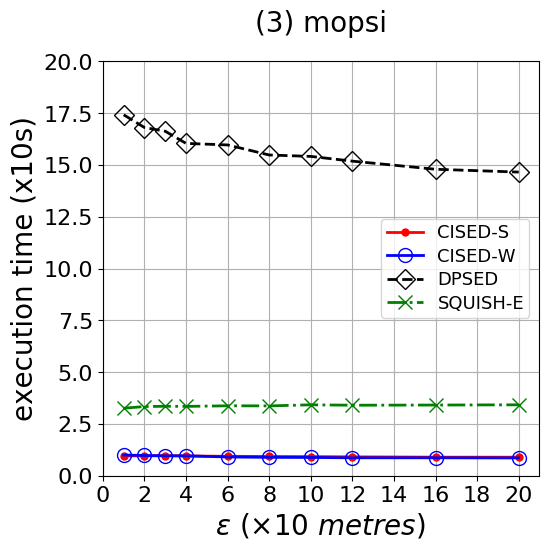}\hspace{1ex}
\includegraphics[scale = 0.290]{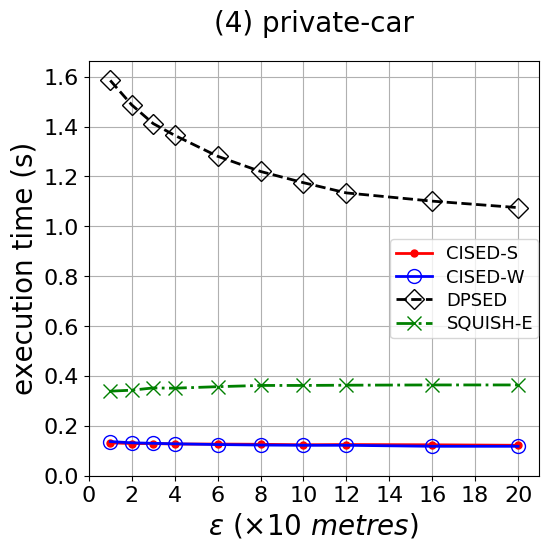}
\caption{\small Evaluation of running time: fixed with $m=16$ and varying error bounds $\epsilon$.}
\label{fig:time-epsilon}
\end{figure*}

\begin{figure*}[tb!]
\centering
\includegraphics[scale = 0.290]{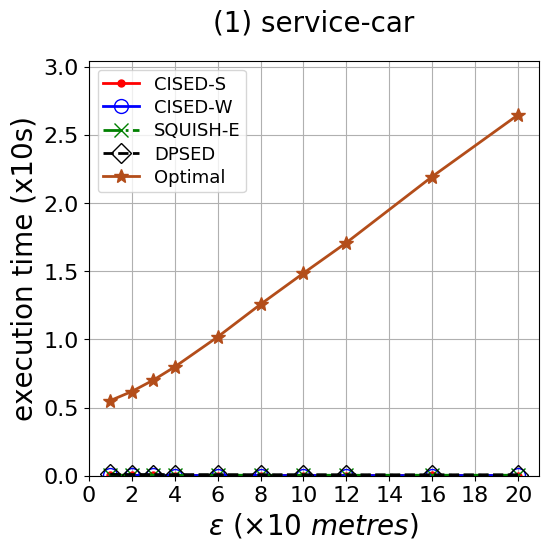}\hspace{1ex}
\includegraphics[scale = 0.290]{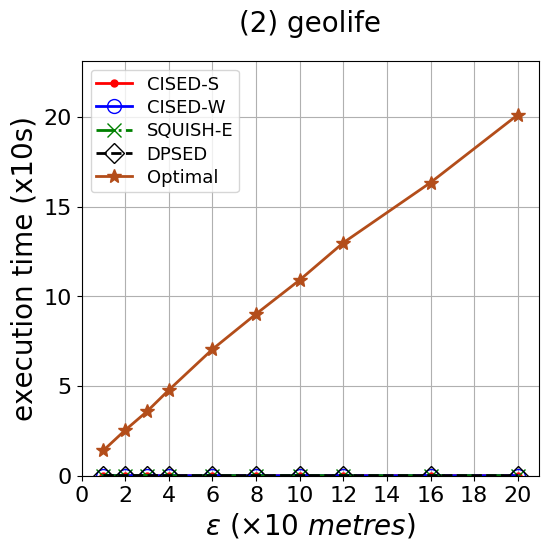}\hspace{1ex}
\includegraphics[scale = 0.290]{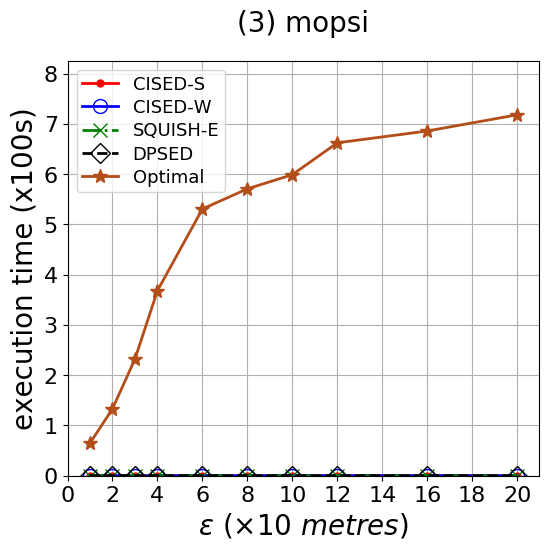}\hspace{1ex}
\includegraphics[scale = 0.290]{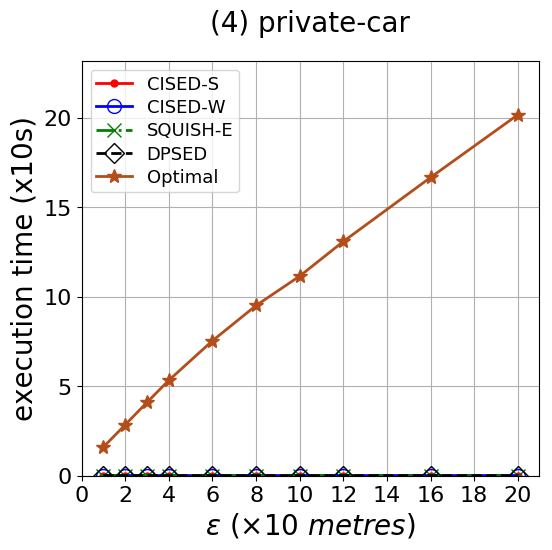}
\caption{\small Evaluation of running time: fixed with $m=16$ and varying error bounds $\epsilon$ (on small datasets).}
\label{fig:time-optimal-epsilon}
\end{figure*}

\begin{figure*}[tb!]
\centering
\includegraphics[scale = 0.290]{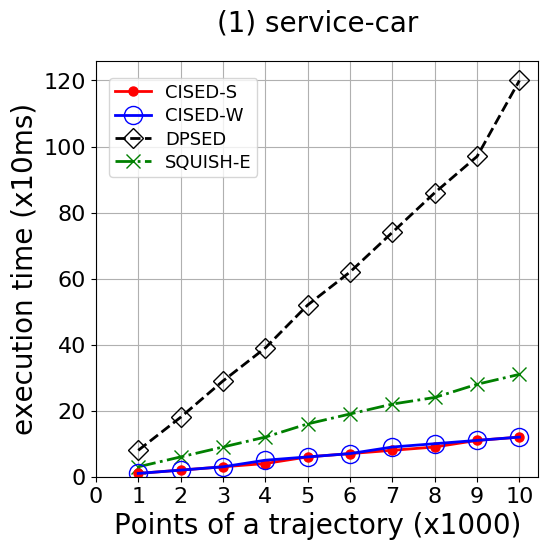}\hspace{1ex}
\includegraphics[scale = 0.290]{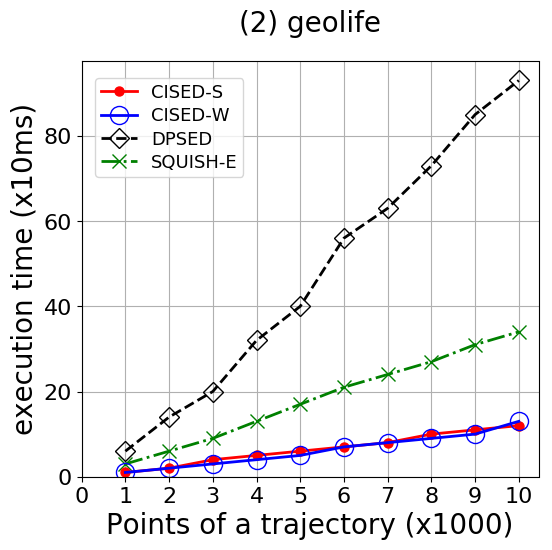}\hspace{1ex}
\includegraphics[scale = 0.290]{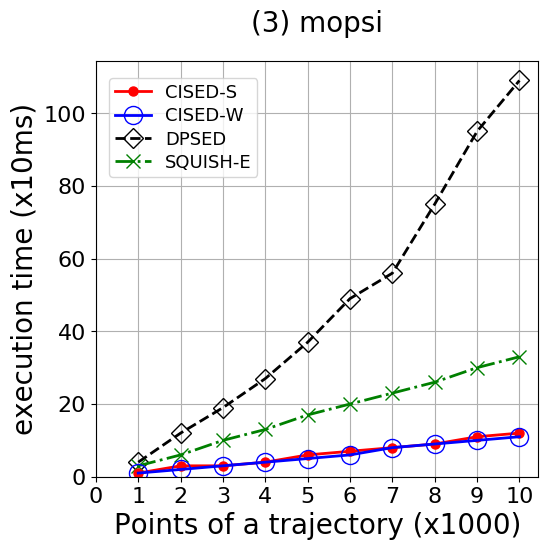}\hspace{1ex}
\includegraphics[scale = 0.290]{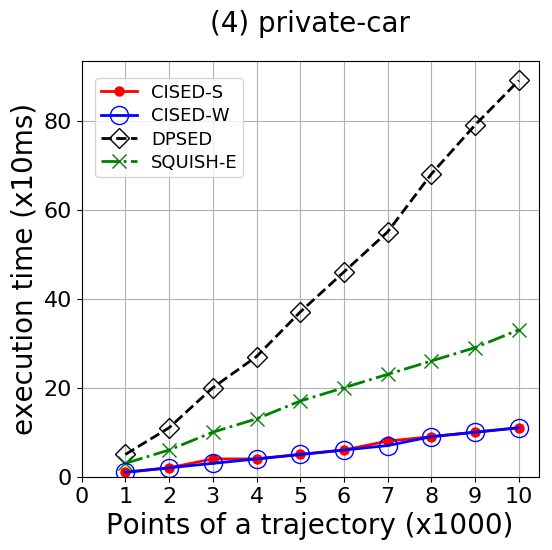}
\caption{\small Evaluation of running time: fixed with $m=16$ and $\epsilon=60$ meters, and varying the size of trajectories. }
\label{fig:time-size}
\vspace{-1ex}
\end{figure*}

\stitle{{Exp-3.2}:  Impacts of the error bound $\epsilon$ on efficiency (VS. algorithms \dps and \squishe).}
To evaluate the impacts of $\epsilon$ on efficiency, we fixed $m$ = $16$,
and varied $\epsilon$  from $10$ meters to $200$ meters on the entire
datasets (\sercar, \geolife, \mopsi, \pricar), respectively.
The results are reported in Figure~\ref{fig:time-epsilon}.

\ni(1) All algorithms are not very sensitive to $\epsilon$ on any datasets, and algorithm \dps is more sensitive to $\epsilon$ than the other three algorithms.
The running time of \dps decreases a little bit with the increase of $\epsilon$, as the increment of $\epsilon$ decreases the number of partitions of the input trajectory.

\ni(2) Algorithms \cist and \cista are obviously faster than \dps and \squishe for all cases.
They are on average ($14.21$, $18.19$, $17.06$, $9.98$) times faster than \dps,
and ($2.84$, $3.45$, $3.69$, $2.86$) times faster than \squishe on
{datasets} (\sercar, \geolife, \mopsi, \pricar), respectively.

\stitle{Exp-3.3: Impacts of the error bound $\epsilon$ on efficiency (VS. the optimal algorithm).}
To evaluate the impacts of $\epsilon$ on efficiency, we once again fixed $m$ = $16$,
and varied $\epsilon$ from $10$ meters to $200$ meters on the first $1K$ points of each trajectory of the selected \textit{small datasets}, respectively.
The results are reported in Figure~\ref{fig:time-optimal-epsilon}.

\ni(1) Algorithms \cist and \cista are obviously faster than the optimal algorithm for all cases.
They are on average ($925.25$, $7888.26$, $40041.59$, $8528.76$) times faster than the optimal algorithm on
datasets (\sercar, \geolife, \mopsi, \pricar), respectively.

\stitle{{Exp-3.4}: Impacts of trajectory sizes on efficiency}.
To evaluate the impacts of trajectory sizes on execution time,
we chose the same {$10$} trajectories, from datasets (\sercar, \geolife, \mopsi, \pricar), respectively,
fixed $m$ = $16$ and $\epsilon = 60$ meters, and varied the size \trajec{|T|} of trajectories from $1K$ points to $10K$ points.
The results are reported in Figure~\ref{fig:time-size}.


\ni(1) Algorithms \cist and \cista are both the fastest \lsa algorithms using \sed,
and are {($8.00$--$10.00$, $5.83$--$8.11$, $4.00$--$9.50$, $5.00$--$8.09$) times faster than \dps,
	and {($2.53$--$3.00$, $2.62$--$3.12$, $2.50$--$3.33$, $2.89$--$3.40$)}} times faster than \squishe on the selected $1K$ to $10K$ points datasets (\sercar,
\geolife, \mopsi, \pricar), respectively.

\ni(2) Algorithms \cist and \cista scale well with the increase of the size of trajectories on all datasets,
and both have a linear running time, while algorithm \dps does not.
This is consistent with their time complexity analyses.

\ni(3) The efficiency advantage of algorithms \cist and \cista increases with the increase of trajectory sizes compared with \dps and \squishe.

\subsubsection{Summary}
From these tests we find the following.

\sstab \emph{(1) Polygon intersection Algorithms}. Algorithm \rpia runs faster than \cpia, and is as effective as \cpia.

\sstab\emph{(2) Parameter $m$}. The compression ratio decreases with the increase of $m$, and the running time increases nearly linearly with the increase of $m$. In practice, the range of $[12, 20]$ is a good candidate region for $m$.

\sstab\emph{(3) Compression ratios}. The optimal \lsa algorithm has the best compression ratios among all strong simplification algorithms. Algorithm \cist is comparable with \dps and algorithm \cista is comparable with the optimal \lsa algorithm.
They are all better than \squishe.
The compression ratios of algorithm \cist, the optimal algorithm and algorithm \cista are on average
($79.3\%$, $71.9\%$, $67.3\%$, $72.7\%$),
{($58.1\%$, $45.1\%$, $39.2\%$, $52.8\%$)} and ($57.7\%$, $53.8\%$, $50.0\%$, $54.6\%$) of \squishe
and ($109.2\%$, $108.0\%$, $111.7\%$, $109.1\%$), {($81.3\%$, $75.5\%$, $72.5\%$, $78.1\%$)} and ($79.5\%$, $81.0\%$, $83.0\%$, $82.0\%$) of \dps on {datasets} (\sercar, \geolife, \mopsi, \pricar), respectively.

\sstab\emph{(4) Average errors}. The average errors of these algorithms from the smallest to the largest are \squishe, \dps, \cist, the optimal \lsa algorithm and \cista. Algorithm \cista has obvious higher average errors than \cist as the former essentially forms spatio-temporal cones with a radius of $\epsilon$.

\sstab\emph{(5) Running time}. Algorithms \cist and \cista are the fastest. They are on average
($14.21$, $18.19$, $17.06$, $9.98$), ($2.84$, $3.45$, $3.69$, $2.86$) and ($925.25$, $7888.26$, $40041.59$, $8528.76$) times faster than \dps, \squishe and the optimal \lsa algorithm on datasets (\sercar, \geolife, \mopsi, \pricar), respectively.
The efficiency advantage of algorithms \cist and \cista also increases  with the increase of the trajectory size.


\section{Related Work}
\label{sec-related}

Trajectory compression algorithms are normally classified into two categories, namely lossless compression and lossy compression\cite{Muckell:Compression}.
(1) Lossless compression methods enable exact reconstruction of the original data from the compressed data without information loss.
(2) In contrast, lossy compression methods allow errors or derivations, compared with the original trajectories.
These techniques typically identify important data points, and remove statistical redundant data points from a trajectory, or replace original data points in a trajectory with other places of interests, such as roads and shops. They focus on good compression ratios with acceptable errors.
%
In this work, we focus on lossy compression of trajectory data, and  we next introduce the related work on lossy trajectory compression  from two aspects: line simplification based methods and semantics based methods.

\subsection{Line simplification based methods}

The idea of piece-wise line simplification comes from computational geometry.
Its target is to approximate a given finer piece-wise linear curve by another coarser piece-wise linear curve, which is typically a subset of the former, such that the maximum distance of the former to the later is bounded by a user specified bound $\epsilon$.
Initially, line simplification (\lsa) algorithms use perpendicular Euclidean distances (\ped) as the distance metric.
Then a new distance metric, the synchronous Euclidean distances (\sed), was developed after the \lsa algorithms were introduced to compress trajectories.
\sed was first introduced in the name of \emph{time-ratio distance} in~\cite{Meratnia:Spatiotemporal}, and formally presented in~\cite{Potamias:Sampling} as the \emph{synchronous Euclidean distance}.
\ped and \sed are two common  metrics adopted in trajectory simplification. The former usually brings better compression ratios while the later reserves temporal information in the result trajectories.

Line simplification algorithms can be classified into two aspects: optimal and sub-optimal methods.

\subsubsection{Optimal Algorithms}
For the ``min-\#" problem that finds out the minimal number of points or segments to represent the original polygonal lines \wrt an error bound $\epsilon$, Imai and Iri \cite{Imai:Optimal} first formulated it as a graph problem, and showed that it could be solved in  $O(n^3)$ time, where $n$ is the number of the original points.
Toussaint of \cite{Toussaint:Optimal} and Melkman and O'Rourke of \cite{Melkman:Optimal} improved the time complexity to $O(n^2 \log n)$ by using either \textit{convex hull} or \textit{sector intersection} methods.
The authors of \cite{Chan:Optimal} further proved that the optimal algorithm using \ped could be implemented in $O(n^2)$ time by using the \textit{sector intersection} mechanism.
Because the \textit{sector intersection} and the \textit{convex hull} mechanisms can not work with \sed, hence, currently the time complexity of the optimal algorithm using \sed remains $O(n^3)$.
It is time-consuming and impractical for large trajectory data~\cite{Heckbert:Survey}.

\subsubsection{Sub-optimal Algorithms}
Many studies have been targeting at finding the sub-optimal results.
In particular, the state-of-the-art of sub-optimal \lsa approaches fall into three categories, \ie batch, online and one-pass algorithms.
We next introduce these \lsa  based trajectory compression algorithms from the aspect of the three categories.

\stitle{Batch algorithms.}
The batch algorithms adopt a global distance checking policy that requires all trajectory points are loaded before compressing starts.
These batch algorithms can be either top-down or bottom-up.

Top-down algorithms, e.g., Ramer \cite{Ramer:Split} and Douglas-Peucker \cite{Douglas:Peucker}, recursively divide a trajectory into sub-trajectories until the stopping condition is met.
Bottom-up algorithms, e.g., Theo Pavlidis' algorithm \cite{Pavlidis:Segment}, is the natural complement of the top-down ones, which recursively merge adjacent sub-trajectories with the smallest distance, initially $n/2$  sub-trajectories for a trajectory with $n$ points, until the stopping condition is met.
The distances of newly generated line segments are recalculated during the process.
These batch algorithms originally only support \ped, but are easy to be extended to support \sed~\cite{Meratnia:Spatiotemporal}.
The batch nature and high time complexities make batch algorithms impractical for online  and resource-constrained scenarios \cite{Lin:Operb}.

\stitle{Online algorithms.}
The online algorithms adopt a constrained global distance checking policy that restricts the checking within a sliding or opening window.
Constrained global checking algorithms do not need to have the entire trajectory ready before they start compressing, and are more appropriate than batch algorithm for compressing trajectories for online scenarios.

Several \lsa algorithms have been developed, e.g., by combining \dpa or \pavlidis with sliding or opening windows for online processing\cite{Meratnia:Spatiotemporal}. 
These methods still have a high time and/or space complexity, which significantly hinders their utility in resource-constrained mobile devices \cite{Liu:BQS}.
\bqsa \cite{Liu:BQS, Liu:Amnesic} and \squishe\cite{Muckell:Compression} further optimize the opening window algorithms.
\bqsa \cite{Liu:BQS, Liu:Amnesic} fasts the processing by picking out at most eight special points from an open window based on a convex hull, which, however, hardly supports \sed.
The \squishe\cite{Muckell:Compression} algorithm is an combination of {opening} window and bottom-up online algorithm. It uses a doubly linked list $Q$ to achieve a better efficiency. Although \squishe supports \sed, it is not one-pass, and has a relatively poor compression ratio.

\stitle{One-pass algorithms.}
The one-pass algorithms adopt a local distance checking policy.
\eat{The local checking policy, the key to achieve the \emph{one-pass} processing, They do not need a window to buffer the preview read points.
Meanwhile, a trajectory compression algorithm is {\em one-pass} if it processes each point in a trajectory once and only once when compressing the trajectory.
}
They do not need a window to buffer the previously read points as they process each point in a trajectory once and only once.
Obviously, the one-pass algorithms run in linear time and constant space.

The $n$--${th}$ point routine and the routine of random-selection of points \cite{Shi:Survey} are two naive one-pass algorithms.
In these routines, for every fixed number of consecutive points along the line, the $n$--${th}$  point and one random point among them are retained, respectively.
They are fast, but are obviously not error bounded.
In Reumann-Witkam routine\cite{Reumann:Strip}, it builds a strip paralleling to the line connecting the first two points, then the points within this strip compose one section of the line.
The Reumann-Witkam routine also runs fast, but has limited compression ratios.
The sector intersection (\cia) algorithm \cite{Williams:Longest, Sklansky:Cone} was developed for graphic and pattern recognition in the late 1970s, for the approximation of arbitrary planar curves by linear segments or finding a polygonal approximation of a set of input data points in a 2D Cartesian coordinate system. \cite{Dunham:Cone} optimized algorithm \cia by considering the distance between a potential end point and the initial point of a line segment, and the \sleeve algorithm \cite{Zhao:Sleeve} in the cartographic discipline essentially applies the same idea as the \cia algorithm.
Moreover, {fast \bqsa \cite{Liu:BQS} (\fbqsa in short), the simplified version of \bqsa, has a linear time complexity.}
The authors of this article also developed an One-Pass ERror Bounded (\operb) algorithm \cite{Lin:Operb}.
However, all existing one-pass algorithms  use \ped \cite{Williams:Longest, Sklansky:Cone, Dunham:Cone, Zhao:Sleeve,Liu:BQS,Lin:Operb}, while this study focuses on \sed.


\subsection{Semantics based methods}
The trajectories of certain moving objects such as cars and trucks are constrained by road networks. These moving objects typically travel along road networks, instead of the line segment between two points. Trajectory compression methods based on road networks \cite{Chen:Trajectory, Popa:Spatio,Civilis:Techniques,Hung:Clustering, Gotsman:Compaction, Song:PRESS, Han:Compress}  project trajectory points onto roads (also known as Map-Matching). Moreover, \cite{Gotsman:Compaction, Song:PRESS, Han:Compress} mine and use high frequency patterns of compressed trajectories, instead of roads, to further improve compression effectiveness.
Some methods \cite{Schmid:Semantic, Richter:Semantic} compress trajectories beyond the use of road networks, and further make use
of other user specified domain knowledge, such as places of interests along the trajectories \cite{Richter:Semantic}.
There are also compression algorithms preserving the direction of the trajectory \cite{Long:Direction}. 

These  semantics based approaches are orthogonal to line simplification based methods, and may be combined with each other to  improve the effectiveness of trajectory compression.

\vspace{-1ex}
\section{Conclusions}  
\label{sec-conclusion}

We have proposed \cist and \cista, two one-pass error bounded strong and weak trajectory simplification algorithms using the synchronous distance.
We have also experimentally verified that algorithms \cist and \cista are both efficient and effective.
They are three times faster than \squishe, the most efficient existing \lsa algorithm using \sed.
In terms of compression ratio,
algorithm \cist is {comparable} with \dps, the existing \lsa algorithm with the best compression ratio, and is $21.1\%$ better than \squishe on average; and
algorithm \cista is comparable with the optimal algorithm and is on average $19.6\%$ and $42.4\%$ better than \dps and \squishe, respectively.


\section*{Acknowledgments}
This work is supported in part by NSFC (U1636210), 973 program ({2014CB340300}), NSFC ({61421003}) and  Beijing Advanced Innovation Center for Big Data and Brain Computing.

\balance
\bibliographystyle{abbrv}
\bibliography{sec-ref}

\end{document}